\documentclass[10pt]{article}

\usepackage{multirow}
\usepackage{booktabs}
\usepackage{tabularx}
\usepackage{ltablex}
\usepackage{array}
\usepackage{tablefootnote}

\usepackage{amssymb,amsmath,amsthm}
\usepackage{nicefrac,bbm}
\theoremstyle{definition}   
\usepackage{apptools,mathtools,chngcntr}

\usepackage{color,xcolor}
\usepackage{tikz,pgfplots}
\usepgfplotslibrary{colormaps,fillbetween,groupplots}
\usepackage{pgfplotstable}
\usepackage{subcaption}
\usepackage[font=small]{caption}
\usepackage{wrapfig}
\pgfplotsset{compat=newest}
\usetikzlibrary{patterns,shapes,arrows.meta,patterns.meta}
\tikzset{>=latex} 
\usepackage{pgfgantt}

\usepackage[maxbibnames=50, maxcitenames=2, uniquelist=false,
    natbib=true, style=authoryear-comp,
    isbn=false, url=false,giveninits=true,dashed=false]{biblatex}   
\bibliography{bibliography}
\AtBeginBibliography{} 

\usepackage[utf8]{inputenc}
\usepackage[T1]{fontenc}
\usepackage[a4paper,left=25mm,right=25mm,top=25mm,bottom=25mm]{geometry}
\usepackage[hidelinks]{hyperref}
\usepackage{pdfpages}
\usepackage[capitalise]{cleveref}
\crefname{section}{Section}{Sections}
\crefname{appendix}{Appendix}{Appendices}
\crefname{table}{Table}{Tables}
\usepackage{titlesec, blindtext}
\usepackage{marginnote}

\newtheorem{definition}{Definition}[section]
\newtheorem{theorem}[definition]{Theorem}
\newtheorem{corollary}[definition]{Corollary}
\newtheorem{remark}[definition]{Remark}

\usepackage[affil-it]{authblk}

\usepackage{enumitem}

\newcommand{\fo}{\mathfrak{o}}
\newcommand{\fp}{\mathfrak{p}}
\newcommand{\fq}{\mathfrak{q}}
\newcommand{\fr}{\mathfrak{r}}
\newcommand{\fs}{\mathfrak{s}}
\newcommand{\ft}{\mathfrak{t}}

\newcommand{\ba}{\boldsymbol{a}}
\newcommand{\bb}{\boldsymbol{b}}

\newcommand{\bd}{\boldsymbol{d}}
\newcommand{\be}{\boldsymbol{e}}

\newcommand{\bu}{\boldsymbol{u}}
\newcommand{\bn}{\boldsymbol{n}}

\newcommand{\bw}{\boldsymbol{w}}
\newcommand{\bx}{\boldsymbol{x}}

\newcommand{\bA}{\boldsymbol{A}}
\newcommand{\bB}{\boldsymbol{B}}
\newcommand{\bC}{\boldsymbol{C}}

\newcommand{\bF}{\boldsymbol{F}}
\newcommand{\bG}{\boldsymbol{G}}
\newcommand{\bH}{\boldsymbol{H}}
\newcommand{\bI}{\boldsymbol{I}}

\newcommand{\bM}{\boldsymbol{M}}

\newcommand{\bP}{\boldsymbol{P}}
\newcommand{\bQ}{\boldsymbol{Q}}

\newcommand{\bW}{\boldsymbol{W}}

\newcommand{\bxi}{\boldsymbol{\xi}}

\newcommand{\cC}{{\mathcal{C}}}
\newcommand{\cD}{{\mathcal{D}}}

\newcommand{\cF}{{\mathcal{F}}}
\newcommand{\cG}{{\mathcal{G}}}

\newcommand{\cI}{{\mathcal{I}}}

\newcommand{\cK}{{\mathcal{K}}}
\newcommand{\cL}{{\mathcal{L}}}
\newcommand{\cM}{{\mathcal{M}}}

\newcommand{\cO}{{\mathcal{O}}}
\newcommand{\cP}{{\mathcal{P}}}

\newcommand{\cS}{{\mathcal{S}}}
\newcommand{\cT}{{\mathcal{T}}}

\newcommand{\bbA}{{\mathbb{A}}}

\newcommand{\bbM}{{\mathbb{M}}}
\newcommand{\bbN}{{\mathbb{N}}}

\newcommand{\bbR}{{\mathbb{R}}}

\newcommand{\bnull}{\boldsymbol{0}}

\newcommand{\bcC}{\hat{\boldsymbol{{c}}}}

\newcommand{\bcI}{\boldsymbol{\mathcal{I}}}

\newcommand{\bcP}{\boldsymbol{\mathcal{P}}}

\newcommand{\RL}{\mathcal{RL}}
\newcommand{\SP}{\mathcal{SP}}

\newcommand{\GL}{\text{GL}}

\newcommand{\SO}{\text{SO}}
\newcommand{\SYM}{\text{SYM}}

\newcommand{\tr}{\operatorname{tr}}

\newcommand{\cof}{\operatorname{cof}}

\newcommand{\norm}[1]{\left\lVert#1\right\rVert}
\newcommand{\normm}[1]{\big\lVert#1\big\rVert}

\newcommand{\Cross}{\mathbin{\tikz [x=1.4ex,y=1.4ex,line width=.25ex] \draw (0.1,0.1) -- (0.9,0.9) (0.1,0.9) -- (0.9,0.1);}}

\newcommand{\bMti}{\bM^{\text{ti}}}
\newcommand{\bMcub}{\bbM^{\text{cub}}}
\newcommand{\bMtet}{\bbM^{\text{tet}}}
\newcommand{\bMmonone}{\bM^{\text{mon}}_{(1)}}
\newcommand{\bMmontwo}{\bM^{\text{mon}}_{(2)}}
\newcommand{\bMmononePC}{\widetilde{\bM}{^{\text{mon}}_{(1)}}}
\newcommand{\bMmontwoPC}{\widetilde{\bM}{^{\text{mon}}_{(2)}}}
\newcommand{\bMrho}{\bM^{\text{rho}}}
\newcommand{\bMrhoONE}{\widetilde{\bM}{^{\text{rho}}_1}}
\newcommand{\bMrhoTWO}{\widetilde{\bM}{^{\text{rho}}_2}}

\newcommand{\JoneISO}{J_1^{\text{iso}}}
\newcommand{\JtwoISO}{J_2^{\text{iso}}}
\newcommand{\JthreeISO}{J_3^{\text{iso}}}

\newcommand{\IoneISO}{I_1^{\text{iso}}}
\newcommand{\ItwoISO}{I_2^{\text{iso}}}
\newcommand{\IthreeISO}{I_3^{\text{iso}}}
\newcommand{\IfourISO}{I_4^{\text{iso}}}

\newcommand{\IoneCUB}{I_1^{\text{cub}}}
\newcommand{\ItwoCUB}{I_2^{\text{cub}}}
\newcommand{\IthreeCUB}{I_3^{\text{cub}}}
\newcommand{\IfourCUB}{I_4^{\text{cub}}}
\newcommand{\IfiveCUB}{I_5^{\text{cub}}}
\newcommand{\IsixCUB}{I_6^{\text{cub}}}

\newcommand{\JoneCUB}{J_1^{\text{cub}}}
\newcommand{\JtwoCUB}{J_2^{\text{cub}}}
\newcommand{\JthreeCUB}{J_3^{\text{cub}}}
\newcommand{\JfourCUB}{J_4^{\text{cub}}}
\newcommand{\JfiveCUB}{J_5^{\text{cub}}}
\newcommand{\JsixCUB}{J_6^{\text{cub}}}

\newcommand{\IoneTET}{I_1^{\text{tet}}}
\newcommand{\ItwoTET}{I_2^{\text{tet}}}
\newcommand{\IthreeTET}{I_3^{\text{tet}}}
\newcommand{\IfourTET}{I_4^{\text{tet}}}
\newcommand{\IfiveTET}{I_5^{\text{tet}}}

\newcommand{\JoneTET}{J_1^{\text{tet}}}
\newcommand{\JtwoTET}{J_2^{\text{tet}}}
\newcommand{\JthreeTET}{J_3^{\text{tet}}}
\newcommand{\JfourTET}{J_4^{\text{tet}}}
\newcommand{\JfiveTET}{J_5^{\text{tet}}}

\newcommand{\IoneTI}{I_1^{\text{ti}}}
\newcommand{\ItwoTI}{I_2^{\text{ti}}}
\newcommand{\IthreeTI}{I_3^{\text{ti}}}
\newcommand{\IfourTI}{I_4^{\text{ti}}}

\newcommand{\JoneTI}{J_1^{\text{ti}}}
\newcommand{\JtwoTI}{J_2^{\text{ti}}}

\newcommand{\IoneMON}{I_1^{\text{mon}}}
\newcommand{\ItwoMON}{I_2^{\text{mon}}}
\newcommand{\IthreeMON}{I_3^{\text{mon}}}
\newcommand{\IfourMON}{I_4^{\text{mon}}}
\newcommand{\IfiveMON}{I_5^{\text{mon}}}

\newcommand{\JoneMON}{J_1^{\text{mon}}}
\newcommand{\JtwoMON}{J_2^{\text{mon}}}
\newcommand{\JthreeMON}{J_3^{\text{mon}}}
\newcommand{\JfourMON}{J_4^{\text{mon}}}
\newcommand{\JfiveMON}{J_5^{\text{mon}}}

\newcommand{\IoneRHO}{I_1^{\text{rho}}}
\newcommand{\ItwoRHO}{I_2^{\text{rho}}}
\newcommand{\IthreeRHO}{I_3^{\text{rho}}}
\newcommand{\IfourRHO}{I_4^{\text{rho}}}

\newcommand{\JoneRHO}{J_1^{\text{rho}}}
\newcommand{\JtwoRHO}{J_2^{\text{rho}}}
\newcommand{\JthreeRHO}{J_3^{\text{rho}}}
\newcommand{\JfourRHO}{J_4^{\text{rho}}}

\newcommand{\Cone}{I_1^{\text{tri}}}
\newcommand{\Ctwo}{I_2^{\text{tri}}}
\newcommand{\Cthree}{I_3^{\text{tri}}}
\newcommand{\Cfour}{I_4^{\text{tri}}}
\newcommand{\Cfive}{I_5^{\text{tri}}}
\newcommand{\Csix}{I_6^{\text{tri}}}
\newcommand{\Cseven}{I_7^{\text{tri}}}
\newcommand{\Ceight}{I_8^{\text{tri}}}
\newcommand{\Cnine}{I_9^{\text{tri}}}
\newcommand{\Cten}{I_{10}^{\text{tri}}}
\newcommand{\Celeven}{I_{11}^{\text{tri}}}
\newcommand{\Ctwelve}{I_{12}^{\text{tri}}}
\newcommand{\Cthirteen}{I_{13}^{\text{tri}}}
\newcommand{\Cfourteen}{I_{14}^{\text{tri}}}

\newcommand{\Kone}{J_1^{\text{tri}}}
\newcommand{\Ktwo}{J_2^{\text{tri}}}
\newcommand{\Kthree}{J_3^{\text{tri}}}
\newcommand{\Kfour}{J_4^{\text{tri}}}
\newcommand{\Kfive}{J_5^{\text{tri}}}
\newcommand{\Ksix}{J_6^{\text{tri}}}

\title{
{
Polyconvex anisotropic hyperelasticity via group symmetrization
}
}

\author[1,*]{Dominik~K.~Klein}
\author[2]{Karl~A.~Kalina}
\author[3]{Rogelio~Ortigosa}
\author[3]{\\Jes\'us~Mart\'inez-Frutos}
\author[2]{Markus~K\"astner}
\author[1]{Oliver~Weeger}
\affil[1]{\footnotesize Cyber-Physical Simulation, 
Department of Mechanical Engineering, TU Darmstadt, 64293 Darmstadt, Germany}
\affil[2]{\footnotesize Institute of Solid Mechanics, TU Dresden, 01062 Dresden, Germany}
\affil[3]{\footnotesize Multiphysics Simulation and Optimization, TU Cartagena, Campus~Muralla~del~Mar, 30202, Cartagena (Murcia), Spain}
\affil[*]{\footnotesize Corresponding author, email: klein@cps.tu-darmstadt.de}

\date{August 12, 2026}
 
\begin{document}

\maketitle 

\vspace{-1ex}
\par\noindent\rule{\textwidth}{0.4pt}

\begin{abstract}

A key challenge in material theory is the formulation of models that satisfy all common mechanical constitutive conditions while retaining sufficient flexibility. In this context, several important modeling aspects remain unresolved for polyconvex anisotropic hyperelasticity. We address some of these challenges and apply our results for physics-augmented neural network (PANN) constitutive modeling. The {\textbf{main contributions}} of this paper are as follows: \textbf{(1)} {We propose a new constitutive modeling framework for polyconvex anisotropic hyperelasticity based on triclinic invariants and group symmetrization, and apply it to a PANN model. For finite symmetry groups, this PANN model fulfills all common mechanical constitutive conditions a priori.} \textbf{(2)} We propose a {new} group symmetrization-based method for the construction of polyconvex invariants for finite symmetry groups. Based on this, we derive a new integrity basis for a tetragonal symmetry group and a new functional basis for a cubic symmetry group. To the best of our knowledge, these are the first polyconvex integrity or functional bases for symmetry groups characterized by structural tensors of order higher than two. \textbf{(3)} We benchmark {different} polyconvex PANN constitutive models with highly nonlinear homogenization data of cubic metamaterials.  {\textbf{(4)} We provide an introduction to the construction of polyconvex invariants, integrity and functional bases and review established polyconvex bases alongside the newly derived ones.}

\end{abstract}
\vspace*{2ex}
{\textbf{Key words:} polyconvexity, hyperelasticity, invariants, integrity basis, functional basis, physics-augmented neural networks, microstructured materials}

\par\noindent\rule{\textwidth}{0.4pt}\vspace*{2pt}
{\small
\noindent Accepted version of manuscript published in the \emph{Journal of the Mechanics and Physics of Solids}.
\noindent Original submission: May 29, 2026. Date accepted: August 12, 2026. DOI: \href{https://doi.org/10.1016/j.jmps.2026.106815}{10.1016/j.jmps.2026.106815}. License: \href{https://creativecommons.org/licenses/by/4.0}{CC BY 4.0}}
\vspace*{-5mm}
\par\noindent\rule{\textwidth}{0.4pt}


\section{Introduction}\label{chap:intro}

Constitutive modeling is one of the pillars of continuum solid mechanics and enables the mathematical description of the behavior of different materials. A fundamental prerequisite of constitutive modeling is the consideration of constitutive conditions, such as thermodynamic consistency and objectivity \parencite{TruesdellNoll}. By including these requirements in the model formulation, mechanically reasonable model predictions can be ensured. Moreover, by prescribing, to some extent, how the material model should behave, these conditions serve as an inductive bias \parencite{haussler1988}. This enables to {train} constitutive models on sparse data while still accurately predicting the material behavior for previously unseen loading conditions. This paradigm has also been adopted in the emerging field of neural network (NN) constitutive modeling, where it is referred to as {thermodynamics-based} \parencite{masi2021}, mechanics-informed \parencite{ASAD2023116463}, physics-based \parencite{aldakheel2025}, {physics-constrained} \parencite{kalina2022b}, {physics-augmented} \parencite{klein2022b}, or constitutive artificial \parencite{Linka2020} NN modeling.\footnote{Throughout this work, we adopt the terminology {physics-augmented neural networks} (PANNs). The paradigm of combining machine learning methods with scientific knowledge is widely used in many scientific fields \parencite{rueden2021,peng2021,karniadakis2021,Kumar2022,kannapinn}.} One important constitutive condition is polyconvexity. Although research on polyconvex constitutive modeling is already highly advanced, several important aspects remain unresolved, some of which we address in this work.


\subsection{Polyconvex constitutive modeling}

In finite elasticity theory, convexity of the potential $W$ in the deformation gradient $\bF$ alone would be overly restrictive and incompatible with a mechanically reasonable material behavior. Specifically, this would be incompatible with growth and objectivity conditions, and for strict convexity, it would fail to capture relevant phenomena like buckling \parencite{ciarlet1988}. In the seminal work by \textcite{Ball1976,Ball1977}, the \textbf{polyconvexity} condition was introduced. Polyconvex potentials are convex in an extended set of arguments, including the deformation gradient $\bF$, its cofactor $\bH=\cof\bF$, and its determinant $J=\det\bF$, making this convexity condition compatible with aforementioned physical considerations. In constitutive modeling, we employ polyconvexity for two reasons. First, it implies ellipticity (or material stability) \parencite{Schroeder_Neff_Balzani_2005}, meaning a stable and robust behavior when applying the constitutive model in numerical applications such as the finite element method. Second, polyconvexity serves as a strong inductive bias that can improve the generalization behavior of the constitutive model, particularly for extrapolation away from the {training} data \parencite{Klein2026a,kalina2024a}. Moreover, from a mathematical perspective, polyconvexity is linked to existence theorems in finite elasticity theory.

\medskip

Formulating material models which are not only polyconvex but also fulfill the remaining constitutive conditions is a very ambitious task. For instance, objective deformation measures easily violate the polyconvexity condition (cf.~\cref{eq:C12}). This does not mean that different constitutive conditions contradict each other, but it illustrates the challenge of fulfilling them all at once. For almost three decades after its initial conception, polyconvex constitutive modeling was practically limited to isotropic material behavior, with models formulated in terms of the main invariants of the right Cauchy-Green tensor $\bC$ or in terms of principal stretches \parencite{ciarlet1988,Ball1976}. In the landmark work of \textcite{Schroeder2003}, for the first time, a polyconvex constitutive model applicable for anisotropic material behavior was proposed, followed by \textcite{Hartmann2003}, \textcite{Itskov2004}, \textcite{Kambouchev2007}, \textcite{Ehret2007}, and \textcite{Schroeder2008}. These approaches are based on polyconvex invariants of $\bC$ and a tuple of structural tensors. Making use of second- and fourth-order structural tensors, \textbf{polyconvex invariants} are constructed for different material symmetry groups. By that, the models fulfill objectivity and material symmetry a priori. Then, by a suitable construction of functions of these invariants, the remaining conditions of hyperelasticity are included in the model, i.e., normalization and growth conditions, while preserving the polyconvexity of the invariants in the model formulations.\footnote{Notably, the polyconvexity condition has inspired a variety of numerical methods~\parencite{schroeder2011,betsch2018,bonet2015,franke2023}, and was also extended to electro-magneto-mechanical material behavior\parencite{gil2016,Ortigosa_Gil_2016_hyperbol,silhavy2018}.}

\medskip

A key prerequisite of polyconvex invariant-based constitutive modeling is the formulation of polyconvex invariants and integrity or functional bases. The latter roughly means that a complete set of invariants is available to fully represent the material behavior for the considered symmetry group in an invariant-based framework. Polyconvex invariants for a variety {of} anisotropies have been proposed in several publications around twenty years ago \parencite{Schroeder2003,Kambouchev2007,Schroeder2008,Schroeder2010a}. To the best of our knowledge, polyconvex integrity or functional bases have been proposed only for the following symmetry groups, all of which can be fully characterized by structural tensors of order up to two \parencite{riemer2025}:\footnote{Throughout this work, we employ the Schoenflies notation to denote material symmetry groups.} the isotropic $\cK$ \parencite{Ball1976}, transversely isotropic $\cD_{\infty}$\parencite{Schroeder2003}, triclinic $\cC_1$ \parencite{Schroeder2008}, monoclinic $\cC_{2}$ \parencite{Schroeder2008}, and rhombic $\cD_{2}$ \parencite{Schroeder2003} groups. In particular, we are not aware of polyconvex integrity or functional bases for symmetry groups that require structural tensors of order higher than two, such as the {tetragonal} $\cD_4$ or the cubic $\cO$ group, which both require structural tensors of third or fourth order \parencite[Tab.~7]{riemer2025}. Thus, for most symmetry groups, polyconvex invariant-based constitutive models are based on incomplete sets of invariants. This leads to a loss of information and limits the flexibility of the models. Moreover, to the best of our knowledge, general construction principles for polyconvex integrity bases have not yet been discussed.


\subsection{PANN constitutive modeling}\label{sec:intro_PANNs}

Shifting the focus to the representation of material models, in recent years, PANN constitutive models have gained considerable attention. NNs are highly flexible functions, in fact, some even have universal approximation properties \parencite{Hornik1991,chen2019optimalcontrolneuralnetworks}. Consequently, using NNs as ansatz functions for the constitutive model equations results in very flexible models that can represent highly nonlinear material behavior.\footnote{Apart from NNs, a variety of sophisticated constitutive modeling approaches has emerged in recent years. For instance, constitutive model equations have been formulated based on Gaussian process regression \parencite{Frankel2020,ELLMER2024116547} or splines \parencite{WIESHEIER2024117208}. Moreover, methods have been proposed to automatically identify suitable constitutive models out of a wide range of predefined models \parencite{flaschel2021,FLASCHEL2026118573,abbasi2026}.} Flexibility is the key feature that differentiates PANN constitutive models from conventional modeling approaches: the former {potentially offer more flexibility} than the latter. By that, PANN constitutive models find application where conventional models tend to fail. First of all, they can accurately represent the highly nonlinear homogenized behavior of microstructured materials, enabling efficient sequential multiscale simulations~\parencite{kalina2022b,gaertner2021,masi2021}. Moreover, PANNs offer a unified modeling framework applicable to a wide range of different material behaviors~\parencite{Peirlinck_Linka_Hurtado_Holzapfel_Kuhl_2025,Klein2026a,JADOON2025117653,FRIEDRICHS2026118852}. This alleviates the challenging process of selecting an appropriate conventional constitutive model from the broad range available~\parencite{Ricker_Wriggers_2023,Steinmann_Hossain_Possart_2012,Hossain_Steinmann_2013}.

\medskip

\textbf{Polyconvex PANN} constitutive models were first introduced by~\textcite{klein2022a}. The approach is based on input-convex neural networks (ICNNs), which were originally proposed in a different context by~\textcite{Amos2017}. Inspired from conventional constitutive models, polyconvex PANN models can be formulated in terms of invariants. By using polyconvex invariants as inputs for a NN with suitable monotonicity and convexity constraints, a polyconvex NN potential can be constructed~\parencite{klein2022a}. This constitutive modeling framework was later extended by invariant-based polyconvex stress normalization terms, resulting in a PANN constitutive model that fulfills all common mechanical conditions of compressible hyperelasticity by construction~\parencite{linden2023}. For isotropy, a variety of polyconvex PANN models based on principal stretches have emerged~\parencite{Geuken_Kurzeja_Wiedemann_Mosler_2025,VIJAYAKUMARAN2024106015,STPIERRE2023116236,mcculloch2024}. Notably, in~\textcite{Geuken_Kurzeja_Wiedemann_Mosler_2025}, a PANN model with universal approximation properties for polyconvex isotropic material behavior is proposed, which is based on the signed singular values of the deformation gradient (which are closely related to principal stretches)~\parencite{wiedemann2023characterizationpolyconvexisotropicfunctions}. Apart from ICNN-based formulations, polyconvex PANN models based on neural ordinary differential equations were proposed \parencite{tac2022}.

While the aforementioned works fulfill polyconvexity in an \emph{exact} fashion, \emph{relaxed} convexity conditions have also been investigated, including strategies such as convexity-promoting loss terms together with a relaxed local version of polyconvexity~\parencite{kalina2024a}, monotonicity of the NN potential~\parencite{Klein2026a}, and convexity with respect to the coordinates of the right Cauchy-Green tensor $\bC$~\parencite{asad2022,Zheng_Kochmann_Kumar_2024}, see~Remark~\ref{rem:conv_C} for a discussion of the latter.
Moreover, polyconvex PANN models have been extended to parametrized~\parencite{Klein_Roth_Valizadeh_Weeger_2023}, multiphysical~\parencite{klein2022b,kalina2024a,Fuhg_Jadoon_Weeger_Seidl_Jones_2024,ORTIGOSA2025117741,HOLTHUSEN2026119082}, and inelastic\footnote{Polyconvexity is related to existence theorems in finite elasticity theory, and, strictly speaking, does not apply to inelastic constitutive models. However, the formal convexity conditions can still be employed for inelastic modeling, e.g., by using energy potentials that are polyconvex functions in the sense of Ball.} material behavior~\parencite{boes2025accountingplasticityextensioninelastic,zlatic2024,KALINA2026118892,HOLTHUSEN2026106337}.

{
Several PANN models have employed the concept of \textbf{group symmetrization} to incorporate the material symmetry condition.\footnote{{In invariant theory, the concept of group symmetrization is referred to as \emph{Reynolds operator} \parencite[Section~3.1.2]{DerksenKemper}.}} For anisotropic hyperelasticity, group symmetrization-based models have been proposed by \textcite{Fernandez2020,klein2022a}. However, the models introduced therein are either objective but not polyconvex, or polyconvex but not objective. Closely related, \textcite{Geuken_Kurzeja_Wiedemann_Mosler_2025} proposed isotropic polyconvex models based on principal stretches or signed singular-values, that include the associated permutation invariance conditions via a group-symmetrization like approach.}

\medskip

To the best of our knowledge, polyconvex \emph{anisotropic} PANN approaches that fulfill all common constitutive conditions a priori are exclusively based on invariants of $\bC$ and structural tensors. While these models perform exceptionally well in a variety of applications \parencite{klein2024a}, for some material behaviors, they tend to fail \parencite{klein2022a,kalina2024a}. For polyconvex \emph{isotropic} hyperelasticity, it has been reported that the latter case can be addressed with alternative model formulations, e.g., using signed singular values instead of main invariants \parencite{Geuken2026}. Even if different polyconvex PANN models are similar in that they satisfy the same constitutive constraints a priori, their capability to represent certain material behaviors can vary significantly \parencite{wollner2026}. This suggests that, also in the anisotropic case, the development of alternative model formulations may lead to an improved predictive performance, and address potential drawbacks of existing approaches.

\subsection{Main contributions and outline}

The main contributions of this work are as follows:\footnote{Parts of the research presented in this manuscript have been published in the first author's doctoral thesis \parencite{kleinDiss}.} \textbf{(1)} {We propose a new constitutive modeling framework for polyconvex anisotropic hyperelasticity based on triclinic invariants and group symmetrization, and apply it to a PANN model. For finite symmetry groups, this PANN model fulfills all common mechanical constitutive conditions a priori.} \textbf{(2)} We propose a {new} group symmetrization-based method for the construction of polyconvex invariants for finite symmetry groups. Based on this, we derive a new integrity basis for a tetragonal symmetry group and a new functional basis for a cubic symmetry group. To the best of our knowledge, these are the first polyconvex integrity or functional bases for symmetry groups characterized by structural tensors of order higher than two. \textbf{(3)} We benchmark {different} polyconvex PANN constitutive models with highly nonlinear homogenization data of cubic metamaterials.  {\textbf{(4)} We provide an introduction to the construction of polyconvex invariants, integrity and functional bases and review established polyconvex bases alongside the newly derived ones.}
 
The outline of the remaining manuscript is as follows. In~\cref{chap:basics}, we introduce the fundamentals of {polyconvex invariant-based hyperelasticity}. {The formulation of polyconvex anisotropic hyperelasticity via group symmetrization, which is the key theoretical novelty of our manuscript, is presented in~\cref{chap:PANN_inv}.} Next, in~\cref{sec:PANNs}, we discuss polyconvex PANN modeling approaches based on invariants constructed via structural tensors, as well as approaches based on triclinic invariants and group symmetrization. Subsequently, in~\cref{chap:application_micro}, we apply the PANN models to homogenization data of cubic microstructured materials. This is followed by the conclusion in~\cref{sec:conc}. The notation employed in this work is summarized in Appendix~\ref{app:notation}. {In Appendix~\ref{app:add_inf_invs}, we review established polyconvex invariant bases.}


\section{Fundamentals of {polyconvex invariant-based} hyperelasticity}\label{chap:basics}

{
We now introduce the fundamentals of polyconvex invariant-based hyperelasticity. In~\cref{sec:const_cond}, the constitutive conditions of hyperelasticity are outlined, followed by an introduction to invariant-based modeling in~\cref{sec:inv_basics}. The construction of polyconvex invariants, integrity and functional bases is discussed in~\cref{sec:inv_conv,sec:general_construction}}

{
\subsection{Constitutive conditions}\label{sec:const_cond}

In hyperelasticity, we formulate a constitutive model that connects the deformation gradient\footnote{{Regarding kinematics, $\bF$ denotes the deformation gradient, $\bH=\cof\bF$ its cofactor, and $J=\det\bF$ its determinant. We denote the right Cauchy-Green tensor by $\bC=\bF^T\cdot\bF$ and its cofactor by $\bG=\cof\bC$.}} $\bF$ with the \textbf{first Piola-Kirchhoff stress tensor} $\bP$ at a material point in a body. This can be done by introducing a \textbf{hyperelastic potential} $W$ and defining the stress at its gradient field:\footnote{Throughout this work, we assume that every energy potential is sufficiently continuously differentiable.}
\begin{equation}\label{eq:PK1}
    W:\GL^+(3)\rightarrow\bbR\,,\quad \bF\mapsto W(\bF)\qquad \text{and}\qquad  \bP=\partial_{\bF}W(\bF)\,.
\end{equation}
Thus, as a first constitutive condition, \textbf{thermodynamic consistency} in the form of the energy conservation is fulfilled a priori. In addition, we consider the following basic constitutive conditions, see e.g. \textcite{Haupt2002,Holzapfel2000}:\footnote{{If required, the energy normalization $W(\bF)\rvert_{\bF=\bI}=0$ and the volumetric coercivity condition $W(\bF)\rightarrow\infty$ as $J\rightarrow\infty$ could also be considered, which would require only minor adaptions in the PANN constitutive model, cf.~\parencite[Rem.~2.1, 2.2]{kleinDiss}. Note that in hyperelasticity, objectivity implies fulfillment of the balance of angular momentum \parencite[Proposition~8.3.2]{Silhavy2014}.}}
\begin{itemize}
\setlength\itemsep{0em}
\item \textbf{objectivity}: $W(\bQ\cdot\bF)=W(\bF)\quad\forall(\bF,\,\bQ)\in\GL^+(3)\times\SO(3)$,
\item \textbf{material symmetry}: $W(\bF\cdot\bQ^T)=W(\bF) \quad\forall(\bF,\,\bQ)\in\GL^+(3)\times\mathcal{G}$, \\ where $\mathcal{G}\subseteq\operatorname{O}(3)$ denotes the symmetry group of the considered material,
\item \textbf{stress normalization:} $ \bP(\bF)\big\rvert_{\bF=\bI}=\bnull$,
\item \textbf{volumetric growth condition:} $W(\bF)\rightarrow\infty$ as $J \rightarrow 0^+$.
\end{itemize}

Further constitutive conditions are grounded in the concept of convexity. In this work, we focus on the \textbf{polyconvexity} condition introduced by \textcite{Ball1976,Ball1977}. A hyperelastic potential is polyconvex if there exists a convex function $\cP$ such that
\begin{equation}\label{eq:pc_mech_comp}
\cP:\cL_2\times\cL_2\times\bbR\rightarrow\bbR\,,\qquad (\bF,\,\bH,\,J)\mapsto\cP(\bF,\,\bH,\,J)\,,
\end{equation}
with $W(\bF)=\cP(\bF,\,\bH,\,J)$. Polyconvexity is linked to existence theorems in finite elasticity theory. From an engineering perspective, the relevance of existence theorems may be questioned, which we further discuss in \textcite[Sec.~1.2]{klein2026b}. We mainly employ polyconvexity because it can improve the model generalization, and because it is the most straightforward way of ensuring the \textbf{ellipticity} (or \textbf{rank-one convexity}) condition \parencite{Zee1983}
\begin{equation}\label{eq:ellip_mech_comp}
\begin{aligned}
(\ba\otimes\bb):\partial^2 _{\bF\bF}W:(\ba\otimes\bb)\geq 0
\qquad \forall \ba,\bb\in\bbR^3\quad \text{such that}\quad\bF+\ba\otimes\bb\in\GL^+(3)\,.
    \end{aligned}
\end{equation}
Ellipticity ensures material stability \parencite{Schroeder_Neff_Balzani_2005}, meaning a stable and robust behavior when applying the constitutive model in numerical applications such as the finite element method. 
}

\subsection{Invariant-based constitutive modeling}\label{sec:inv_basics}

In invariant-based modeling, the hyperelastic potential does not directly depend on the deformation gradient $\bF$. Instead, the potential is a function of a tuple of invariants $ \bcI^{\square}$, i.e.,
\begin{equation}\label{eq:pot_inv_basics}
    \psi:\bbR^m\rightarrow\bbR\,,\quad \bcI^{\square}\mapsto \psi(\bcI^{\square})\qquad \text{and}\qquad\bP=\partial_{\bF}W(\bF)=\sum_{\alpha=1}^m\partial_{\cI^{\square}_{\alpha}}\psi\,d_{\bF}\cI^{\square}_{\alpha}\,,
\end{equation}
with $W(\bF)=\psi(\bcI^{\square})$, and where $\square$ denotes the considered symmetry group. {By this, material symmetry and objectivity can be fulfilled by construction \parencite{Holzapfel2000}. When calculating the first Piola-Kirchhoff stress in an invariant-based setting,} the derivatives of the potential w.r.t. the invariants are referred to as \emph{stress coefficients}, while the derivatives of the invariants w.r.t.~the deformation gradient are referred to as \emph{tensor generators} \parencite{kalina2022}. {So far, two ways of constructing invariants have been established in the literature:
\begin{enumerate}
\setlength\itemsep{0.0em}
    \item in coordinate-dependent form, where invariants are constructed for a specific choice of reference frame \parencite{smith1962},
    \item more commonly nowadays, based on structural tensors and isotropic tensor functions \parencite{riemer2025,LOKHIN1963597,zhang1990}, which we briefly introduce in~Appendix~\ref{app:invs_established_struct}.
\end{enumerate}
In addition, we propose a new method of constructing invariants in~\cref{sec:inv_bases_group_symm}, and show that
\begin{enumerate}[start=3]
\setlength\itemsep{0.0em}
\item invariants for finite symmetry groups can be constructed via group symmetrization.\footnote{{As discussed in~\cref{sec:intro_PANNs}, group symmetrization is a classical concept of invariant theory, and it has already been employed in constitutive modeling, albeit not in the form proposed in our manuscript.}}
\end{enumerate}
}

Let $\bcI^{\square}\in\bbR^m$ be an $m$-tuple containing invariants. We call the tuple of invariants an \emph{integrity basis} when any {polynomial} invariant of the considered symmetry group can be represented as a polynomial in the invariants contained in $\bcI^{\square}$. We call an integrity basis \emph{minimal}\footnote{Traditionally, the mechanics community employed the terminology ``irreducible'' instead of ``minimal'', see e.g. \parencite{linden2023,Ebbing2010}. However, from the mathematical theory of invariants, ``minimal'' is the correct terminology, {cf.~\textcite[Def.~4.1]{olive2017} and~\textcite[Rem.~1]{riemer2025}.}}  when none of the invariants contained in $\bcI^{\square}$ can be represented as a polynomial in the other invariants. Further, we call the tuple of invariants a \emph{functional basis} when any invariant of the considered symmetry group can be represented as a function of the invariants contained in $\bcI^{\square}$. Every integrity basis is a functional basis but not vice versa. For a comprehensive introduction to these concepts {and more rigorous definitions}, the reader is referred to~\textcite[App.~A]{riemer2025}.  

{
\begin{remark}\label{rem:diff_inv_deriv}
While the three methods of constructing invariants differ in their technical derivation, they share several important features. First, the invariant sets obtained by different construction methods are interchangeable, cf.~\parencite[Rem.~3]{riemer2025} and~Appendix~\ref{app:inv_relations}. Second, in classical invariant-based modeling, the invariant set $\bcI^{\square}$ in~\cref{eq:pot_inv_basics} directly encodes the material (an-)isotropy. This is in contrast to the modeling framework based on triclinic invariants and group symmetrization introduced in~\cref{sec:aniso_hyperel_tric_group}, where the hyperelastic potential initially depends only on triclinic invariants, and the desired material symmetry is included in a subsequent group-symmetrization step rather than through the invariant set itself.
\end{remark}
}

\subsection{{Sufficient conditions for polyconvexity of invariant-based potentials}}\label{sec:inv_conv}

Next, we study {polyconvexity of invariant-based potentials}. Let $\bxi:=(\bF,\,\bH,\,J)\in\bbR^{19}$ be the arguments of the polyconvexity condition, cf.~\cref{eq:pc_mech_comp}. Polyconvexity of the potential in~\cref{eq:pot_inv_basics} is equivalent to the positive semi-definiteness (p.s.d.) of the Hessian
\begin{equation}\label{eq:hess_inv}
d^2_{\bxi\bxi}\psi(\bcI^{\square})=[\mathbb{H}_{\psi}]=\underbrace{\left(\partial_{\bxi}\bcI^{\square}\right)^T\cdot\partial^2_{\bcI^{\square}\bcI^{\square}}\psi\cdot\left(\partial_{\bxi}\bcI^{\square}\right)}_{\text{constitutive type term}}\,\,+\underbrace{\partial_{\bcI^{\square}}\psi\cdot\partial^2_{\bxi\bxi}\bcI^{\square}}_{\text{geometric type term}}\succeq 0\,.
\end{equation}
Since, in general, the invariants contained in $\bcI^{\square}$ are nonlinear functions of $\bF,\,\bH,$ and  $J$, the Hessian operator $[\mathbb{H}_{\psi}]$ consists of two terms. The first term includes second derivatives of the potential w.r.t.\ the invariants, which suggests to phrase it as a ``constitutive type term''. The second term includes first derivatives of the potential w.r.t.\ the invariants, which suggests to phrase it as a ``geometric type term'' \parencite{poya2024}. In constitutive modeling, we commonly apply the condition that both terms are p.s.d. \parencite{klein2022a,Schroeder2003}, which is only sufficient but not necessary for the overall p.s.d. of $[\mathbb{H}_{\psi}]$. This leads to the following conditions: When the potential $\psi$ is convex in the invariants $\bcI^{\square}$, the constitutive type term in~\cref{eq:hess_inv} is p.s.d.\ \parencite[Observation~7.1.8]{horn2013}. When the potential $\psi$ is monotonic\footnote{In this work, if not stated otherwise, ``monotonic'' refers to component-wise monotonically increasing functions, i.e., $\partial_{x_i}f(\bx)\geq 0$.} in the invariants $\bcI^{\square}$ and all invariants contained in $\bcI^{\square}$ are polyconvex, the geometric type term in~\cref{eq:hess_inv} is p.s.d. We call an invariant polyconvex if it fulfills the polyconvexity condition. We denote general invariants that might be polyconvex or not by $J^{\square}_i$ and polyconvex invariants by $I^{\square}_i$. For instance, the isotropic invariant $\JoneISO = \norm{\bF}^2$ is convex in $\bF$ and thus polyconvex. However, not all invariants are polyconvex, e.g., the transversely isotropic invariant $\JtwoTI=\tr\bC^2\cdot\bMti$ is not elliptic and thus not polyconvex \parencite{merodio2006}. For unrestricted, non-polyconvex models, also non-polyconvex invariants can be applied, and the potential $\psi$ is not subject to convexity or monotonicity constraints.
{
\begin{remark}
The conditions introduced above are \emph{sufficient but not necessary} for polyconvexity of invariant-based potentials. In particular, we impose the restriction that both individual terms in~\cref{eq:hess_inv} are p.s.d., whereas polyconvexity only requires their sum to be p.s.d. Consequently, polyconvex constitutive models in the considered framework cannot represent every polyconvex function. We further discuss such limitations of polyconvex constitutive modeling in~\textcite{klein2026b}. Note that, for isotropic material behavior, polyconvex PANN models based on signed singular values have been proposed by \textcite{Geuken_Kurzeja_Wiedemann_Mosler_2025}, which have universal approximation properties. The latter model is based on recent theoretical results by~\textcite{wiedemann2023characterizationpolyconvexisotropicfunctions}.
\end{remark}
}


\subsection{Construction of polyconvex integrity and functional bases}\label{sec:general_construction}

When constructing invariant {sets}, the polyconvexity condition is {usually} not considered, and the obtained invariants are generally not polyconvex. Thus, the {obtained} integrity bases have to be adapted to {construct} polyconvex integrity or functional bases applicable for polyconvex constitutive modeling. In the following, we adopt the following scheme to construct polyconvex bases for different symmetry groups:
\begin{enumerate}[label=(\Alph*)]
\setlength\itemsep{0.0em}
    \item Find the (generally non-polyconvex) \textbf{minimal integrity basis} for the considered symmetry group. Minimal integrity bases constructed via structural tensors for all commonly considered symmetry groups are given in~\textcite{riemer2025}.\footnote{Note that corresponding minimal integrity bases can also be formulated using coordinate-based invariants \parencite{smith1962}.}
    \item Try to formulate a \textbf{polyconvex integrity or functional basis}. For this, a polyconvex representation must be found for all of the invariants obtained in~(A). To anticipate~\cref{app:complete_pc_bases}, formulating polyconvex bases remains an open challenge for most anisotropies. 
    \item If required, \textbf{extend the polyconvex integrity or functional basis}. This can either facilitate the formulation of polyconvex constitutive models, or it can improve the model's flexibility.\footnote{{We provide examples for this in~\cref{{eq:K_basis_J},eq:invs_ti_2}.}}
\end{enumerate}
{In hyperelasticity, 14 symmetry groups are commonly considered: eleven crystal symmetry groups, two transversely isotropic symmetry groups, and one isotropic symmetry group. To the best of our knowledge, polyconvex invariants have only been proposed for 10 and polyconvex integrity or functional bases have been proposed for 5 of these symmetry groups, which we summarize in~\cref{tab:invs}. While polyconvex invariants have been proposed for symmetry groups that require up to fourth-order structural tensors, e.g., for the cubic $\cO$ group, polyconvex integrity or functional bases are only available for symmetry groups which can be fully characterized by structural tensors of up to order two, e.g., for the transversely isotropic $\cD_{\infty}$ group. For the convenience of the reader, we provide established polyconvex bases in~\cref{app:add_inf_invs}. Moreover, in~\cref{sec:inv_bases_group_symm}, we introduce a new polyconvex integrity bases for the tetragonal $\cD_4$ group and a new functional basis for the cubic $\cO$ group.}

\begin{remark}
    The transversely isotropic $\cC_{\infty}$ group and the hexagonal $\cC_6$ and $\cD_6$ groups are subgroups of the transversely isotropic $\cD_{\infty}$ group \parencite{Zheng1993}. Thus, they can be partly characterized by the latter \parencite{Schroeder2008}. However, the resulting tuples of invariants remain invariant under a larger set of rotations than those admissible for the respective symmetry groups, i.e., the tuples of employed structural tensors do not satisfy the second condition in \cref{eq:def_struct}. Thus, polyconvex invariants that exactly represent these groups are not available. 
\end{remark}
\begin{remark}
For all common symmetry groups, (generally non-polyconvex) integrity and in turn functional bases are available. For polyconvexity, however, it may not always be possible to formulate integrity bases. The additional flexibility provided by functional bases, i.e., that they admit functional relationships beyond polynomials, can then become necessary, cf.~Appendix~\ref{app:inv_relations_cub}. In practical constitutive modeling, however, the distinction between the two is often of small relevance, as the constitutive model equations are not restricted to polynomials.
\end{remark}

 \renewcommand{\arraystretch}{1.4}
\begin{table}[t!]
 \centering
  \caption{Availability of polyconvex invariants and integrity or functional bases for the most relevant symmetry groups in hyperelasticity. For the $\cK,\,\cD_{\infty},\cC_1,\,\cC_2,\,\cD_2$, and $\cD_4$ groups, integrity bases are available, while for the $\cO$ group, only a functional basis is available. 
}
 {\small
 \begin{tabular}{llp{0.0cm}lp{0.0cm}l}
\toprule
\multicolumn{2}{l}{Symmetry group with}& &Polyconvex invariants && Polyconvex integrity \\[-5pt]
\multicolumn{2}{l}{Schoenflies notation}& & && or functional basis \\
\midrule
\midrule
isotropic & $\cK$ && \textcite{Ball1976} && \textcite{Ball1976} \\
\midrule
transverse isotropy & $\cC_{\infty}$ && n/a && n/a \\
 & $\cD_{\infty}$ && \textcite{Schroeder2003}  && \textcite{Schroeder2003} \\
\midrule
triclinic & $\cC_1$ &&  \textcite{Schroeder2008} && \textcite{Schroeder2008} \\
\midrule
monoclinic & $\cC_2$ && \textcite{Schroeder2008} && \textcite{Schroeder2008} \\
\midrule 
rhombic & $\cD_2$ &&  \textcite{Schroeder2003} && \textcite{Schroeder2003} \\
\midrule
tetragonal & $\cC_4$ && \textcite{Schroeder2010a} && n/a \\
& $\cD_4$ &&  \textcite{Schroeder2010a} && our manuscript, \cref{sec:invs_TET} \\
\midrule
trigonal & $\cC_3$ &&  \textcite{Schroeder2010a} && n/a \\
& $\cD_3$ && \textcite{Schroeder2010a} && n/a \\
		 \midrule
hexagonal & $\cC_6$ && n/a && n/a \\
& $\cD_6$ &&  n/a && n/a \\
		 \midrule
cubic & $\cT$ &&  n/a && n/a \\
	& $\cO$ && \textcite{Kambouchev2007} && our manuscript, \cref{sec:invs_CUB}\\
\bottomrule
\end{tabular}
}
\label{tab:invs}
\end{table}


\section{{Polyconvex anisotropic hyperelasticity via group symmetrization}}\label{chap:PANN_inv}

In this section, {we present the key theoretical novelties of our manuscript. These are based on group symmetrization, which we introduce in~\cref{sec:group_symm_intro}. In~\cref{sec:aniso_hyperel_tric_group}, we introduce a novel polyconvex constitutive modeling framework for finite symmetry groups based on triclinic invariants and group symmetrization. In~\cref{sec:inv_bases_group_symm}, we introduce a new method for the construction of polyconvex invariants based on group symmetrization.}


{
\subsection{Group symmetrization}\label{sec:group_symm_intro}

Group symmetrization is by no means a new idea. Mathematically, it corresponds to the classical group-averaging (Reynolds) operator from invariant theory \parencite[Section~3.1.2]{DerksenKemper}. Moreover, as discussed in~\cref{sec:intro_PANNs}, it has already been employed in PANN constitutive modeling \parencite{Fernandez2020,klein2022a,Geuken_Kurzeja_Wiedemann_Mosler_2025}, albeit not in the form proposed in our manuscript.

The key theoretical novelties of our manuscript are based on the following theorem:
\begin{theorem}\label{th:group_symm}
Let
\begin{equation}
\widetilde{f}:\bbR^m\rightarrow\bbR\,, \quad   \bcC\mapsto \widetilde{f} ( \bcC)\,,
\end{equation}
be a polyconvex function, where the elements of the tuple $\bcC(\bC)\in\bbR^m$ are functions of the right Cauchy-Green tensor $\bC$. For finite symmetry groups $\cG\subset O(3)$ with $\lvert\mathcal{G}\rvert$ number of elements, the group symmetrization\footnote{{We introduce the operator $\star$ in~\cref{eq:def_struct}.}}
\begin{equation}\label{eq:group_symm_intro}
f(\bC)=\frac{1}{\lvert\mathcal{G}\rvert}\sum_{\bQ\in\mathcal{G}}\widetilde{f}\big( \bcC(\bQ\star\bC)\big)\,,
\end{equation}
provides a polyconvex and objective function $f$ that satisfies the material symmetry condition a priori.
\end{theorem}
\begin{proof}
As discussed in~\textcite[Section~3.2]{klein2022a}, $f$ is polyconvex as the group symmetrization operation preserves the polyconvexity of $\widetilde{f}$. Moreover, $f$ is objective as it depends on the right Cauchy-Green tensor. A proof that $f$ satisfies the material symmetry condition for finite symmetry groups can be found in~\textcite[App.~A]{Fernandez2020}.
\end{proof}
\begin{remark}\label{rem:group_symm_O_SO_3}
It is sufficient to consider $\cG\subset\SO(3)$ instead of $\cG\subset\operatorname{O}(3)$ for the evaluation of~\cref{eq:group_symm_intro}. Following \textcite[Prop.~1]{riemer2025}, this is because the elements collected in $\bcC$ are functions of $\bC$, and for functions expressed in $\bC$, a rotation $\bQ\in\SO(3)$ and the corresponding rotoinversion $-\bQ$ have the same effect, see e.g. \textcite[Section~3.7]{Ebbing2010} or \textcite[Section~2.3]{riemer2025}.
\end{remark}
\begin{remark}\label{rem:finite_infinite}
For \emph{infinite} symmetry groups such as the transversely isotropic $\mathcal{D}_{\infty}$ group, the material symmetry condition cannot be fulfilled exactly with the group symmetrization in~\cref{eq:group_symm_intro}. Instead, the material symmetry condition may be approximated by performing the group symmetrization with a finite subgroup of the infinite group \parencite{Fernandez2020}. Consequently, in contrast to the case of \emph{finite} symmetry groups, the resulting constitutive model satisfies the material symmetry condition only approximately. To the best of our knowledge, the application of group symmetrization to infinite material symmetry groups has only been considered in~\textcite[Sec.~5]{klein2022a}, where the transversely isotropic $\mathcal{D}_{\infty}$ group is approximated by six equidistant rotations around the preferred material direction. A systematic investigation of the selection of suitable finite subgroups and its influence on the learned constitutive model remains an open topic for future research.
\end{remark}

\subsection{Polyconvex constitutive modeling for finite symmetry groups based on triclinic invariants and group symmetrization}\label{sec:aniso_hyperel_tric_group}

We now introduce a new polyconvex constitutive modeling framework for finite symmetry groups, which is based on triclinic invariants and group symmetrization. The development of such a model has two main motivations. First, constitutive models are often formulated in terms of structural tensor-based invariants (cf.~\cref{sec:PANN_comp}). However, employing these invariants in polyconvex constitutive modeling can significantly limit the flexibility of such models. In particular, formulating polyconvex integrity or functional bases is highly nontrivial and remains an open challenge for most symmetry groups (cf.~\cref{app:complete_pc_bases}). Thus, a potential based on these invariants cannot represent every function associated with the considered anisotropy, and such models may fail to represent certain material behaviors, even if, in principle, a suitable polyconvex potential might exist. Second, to anticipate~\cref{chap:application_micro}, even if a polyconvex integrity or functional basis is available for the considered symmetry group, in practice, it can be advantageous to employ an alternative polyconvex constitutive model formulation, see also~\parencite{wollner2026,klein2026b}. 

We start with the hyperelastic potential
\begin{equation}\label{eq:pot_c_0}
  \psi:\bbR^{14}\rightarrow\bbR\,, \quad   \bcI^{\text{tri}}\mapsto \psi ( \bcI^{\text{tri}})\,,
\end{equation}
with $W(\bF)=\psi(\bcI^{\text{tri}})$ based on the triclinic integrity basis in coordinate-dependent form\footnote{{Likewise, a triclinic integrity basis constructed via structural tensors and isotropic tensor functions could be employed, cf.~Remark~\ref{rem:diff_inv_deriv} and~\cref{sec:invs_tri}.}}
\begin{equation}\label{eq:K_basis_recall}
\begin{aligned}
  \bcI^{\text{tri}}=\big(&
  {\mathrlap{C_{11},}\phantom{G_{11},}}\,
  {\mathrlap{C_{22},}\phantom{G_{22},}}\,
  {\mathrlap{C_{33},}\phantom{G_{33},}}\,
  {\mathrlap{(C_{11}+C_{22}+2C_{12}) /{2},}\phantom{(G_{11}+G_{22}+2G_{12})/{2},}}\,
  {\mathrlap{(C_{11}+C_{33}+2C_{13})/{2},}\phantom{(G_{11}+G_{33}+2G_{13})/{2},}}\,
  {\mathrlap{(C_{22}+C_{33}+2C_{23})/{2},}\phantom{(G_{22}+G_{33}+2G_{23})/{2},}}\,
        \\
&{\mathrlap{G_{11},}\phantom{G_{11},}}\,
  {\mathrlap{G_{22},}\phantom{G_{22},}}\,
  {\mathrlap{G_{33},}\phantom{G_{33},}}\,
  {\mathrlap{(G_{11}+G_{22}+2G_{12})/{2},}\phantom{(G_{11}+G_{22}+2G_{12})/{2},}}\,
  {\mathrlap{(G_{11}+G_{33}+2G_{13})/{2},}\phantom{(G_{11}+G_{33}+2G_{13})/{2},}}\,
  {\mathrlap{(G_{22}+G_{33}+2G_{23})/{2},}\phantom{(G_{22}+G_{33}+2G_{23})/{2},}}\,
   \\
    &J,\,-J\big)\in\mathbb{R}^{14}\,.
    \end{aligned}
\end{equation}
A polyconvex integrity basis for the triclinic symmetry group $\cC_1$ was first proposed by \textcite[Sec.~3.2]{Schroeder2008}. For the convenience of the reader, we provide the derivation of $\bcI^{\text{tri}}$ in~\cref{app:invs_tri_coord}, and now briefly recall its properties. The first row of~\cref{eq:K_basis_recall} is a minimal polyconvex integrity basis for triclinic anisotropy. We extended this polyconvex integrity basis by polyconvex triclinic invariants of $\bG$, as well as $J$, aiming to improve the expressiveness of the potential in~\cref{eq:pot_c_0}. Including invariants of $\bC$, $\bG$, and $J$, every element of~\cref{eq:K_basis_recall} is objective. Thus, the potential in~\cref{eq:pot_c_0} is also objective. Since $\bcI^{\text{tri}}$ is an integrity basis for triclinic anisotropy and triclinic anisotropy is the most general form of anisotropy, the potential in~\cref{eq:pot_c_0} can represent arbitrary (an-)isotropic material behavior.\footnote{This can be beneficial in scenarios where no information about the material symmetry of the considered material is available, as the constitutive model can learn the material symmetry behavior through the {training} data. Note that several methods have been proposed to identify the unknown symmetry group of a material, e.g., a special PANN architecture~\parencite{kalina2025,Linka2020,Fuhg2022b} and a projection approach~\parencite{ellmer2025}.}

As discussed in~Theorem~\ref{th:group_symm}, for finite symmetry groups, the group symmetrization
\begin{equation}\label{eq:group_symm}
\widetilde{\psi}(\bC)=\frac{1}{\lvert\mathcal{G}\rvert}\sum_{\bQ\in\mathcal{G}}{\psi}\big( \bcI^{\text{tri}}(\bQ\star\bC)\big)\,,
\end{equation}
provides a potential $\widetilde{\psi}$ that fulfills material symmetry and objectivity a priori. All elements of $\bcI^{\text{tri}}$ are polyconvex, and, consequently, polyconvexity of $\psi$ can be ensured by formulating it as a convex and monotonic function (cf.~\cref{sec:inv_conv}). Note that the inclusion of the remaining constitutive constraints, namely stress normalization and the volumetric growth condition (cf.~\cref{sec:const_cond}), require further adaptions of the modeling framework, for which one possible realization will be presented in~\cref{sec:PANN_C}.
\begin{remark}\label{rem:comput_cost}
    For the group symmetrization in~\cref{eq:group_symm}, multiple evaluations of $\psi$ are required, leading to an increased computational cost compared to approaches that employ invariants constructed via structural tensors and isotropic tensor functions, cf.~\cref{app:invs_established_struct} and~Remark~\ref{rem:diff_inv_deriv}.
\end{remark}
\begin{remark}
In recent literature, related modeling frameworks are sometimes referred to as \textit{coordinate-based modeling} \parencite{kleinDiss}. This is motivated by the non-polyconvex triclinic integrity basis in~\cref{eq:K_basis_unr}, which so far has been used for these models and includes all six independent coordinates of the right Cauchy-Green tensor. Notably, this terminology is not motivated by the formulation of invariants in coordinate-dependent form, cf.~\parencite{smith1962,riemer2025}. The latter would also be possible for other symmetry groups, e.g., transverse isotropy \parencite[Sec.~2.2]{schroeder2010b} or the eleven groups arising from the crystal systems \parencite{smith1962}.
\end{remark}
\begin{remark}\label{rem:conv_C}
Some works employ a similar convexity condition, where the hyperelastic potential is convex in the non-polyconvex triclinic integrity basis (cf.~\cref{eq:invs_ti_unr}), i.e., in the six independent coordinates of $\bC$ \parencite{Zheng_Kochmann_Kumar_2024,asad2022}. This does not ensure polyconvexity, which becomes evident from~\cref{eq:hess_inv}. While for such models, the constitutive type term is p.s.d., the geometric type term can take negative eigenvalues as (i) the potential is not a monotonically increasing function and (ii) not all coordinates of $\bC$ are polyconvex, cf.~Appendix~\ref{app:invs_tri_coord}. Still, convexity of the potential in $\bC$ can be seen as a \emph{relaxed} polyconvexity condition. Importantly, convexity of the potential in $\bC$ does not imply ellipticity \parencite{gao2017}.
\end{remark}
}

\subsection{{Construction of polyconvex invariants via group symmetrization}}\label{sec:inv_bases_group_symm}

{We now introduce a new method for the construction of polyconvex invariants, which is again based on group symmetrization. We exemplarily apply our method for the construction of a new polyconvex integrity bases for the tetragonal (tet) $\cD_4$ group and a new functional basis for the cubic (cub) $\cO$ group, which we abbreviate as indicated in the brackets. To the best of our knowledge, these are the first polyconvex integrity or functional bases for symmetry groups that require structural tensors of order higher than two (cf.~\cref{sec:general_construction}). For the convenience of the reader, we provide established polyconvex bases for several symmetry groups in~\cref{app:add_inf_invs}.

We denote the preferred direction of a material in the reference configuration by orthonormal vectors
\begin{equation}\label{eq:pref_dir}
\bn_{0}^{(i)}\in\bbR^3\qquad
\text{with}\qquad\bn_{0}^{(i)}\cdot\bn_{0}^{(j)}=\delta_{ij}\,.
\end{equation} 
Following the procedure outlined in~\cref{sec:general_construction}, we start with the (generally non-polyconvex) integrity bases provided in~\textcite{riemer2025}, which are derived via the concept of structural tensors and isotropic tensor functions. For an extensive introduction to structural tensor-based modeling, the reader is referred to \textcite{riemer2025}. We provide a brief introduction to this concept in~Appendix~\ref{app:invs_established_struct}.}


\subsubsection{Tetragonal material behavior}\label{sec:invs_TET}

Two tetragonal symmetry groups are relevant in anisotropic hyperelasticity \parencite{riemer2025}. In this work, we only consider the tetragonal $\cD_4$ group.

\paragraph{(A) Minimal integrity basis}

For the tetragonal symmetry group $\cD_4$, a single fourth-order structural tensor is sufficient to characterize the material's anisotropy, e.g.,
\begin{equation}\label{eq:tet_struct}
   \cS^{\text{tet}}=\big(\bMtet\big)\qquad\text{with}\qquad \bMtet=\sum_{i=1}^2 \bn_{0}^{(i)}\otimes\bn_{0}^{(i)}\otimes\bn_{0}^{(i)}\otimes\bn_{0}^{(i)}\,.
\end{equation}
Together with the isotropic invariants $J^{\text{iso}}_{1-3}$ introduced in~\cref{sec:invs_iso}, the invariants
\begin{equation}\label{eq:invs_tet_brain}
\begin{aligned}
&{\mathrlap{\JoneTET = \tr\bMtet:\bC\,}\phantom{\JoneTET = \tr\bMtet:\bC\,}}\qquad 
    {\mathrlap{\JtwoTET = \tr\left(\bMtet:\bC\right)^2}{\phantom{\JfourTET = \tr\bC^2\cdot\left(\bMtet:\bC\right)}}}\,,
    \qquad
 \JthreeTET = \tr\bMtet:\bC^2\,,
    \\ 
    &\phantom{\JoneTET = \tr\bMtet:\bC\,}\qquad \JfourTET = \tr\bC^2\cdot\left(\bMtet:\bC\right)\,,\qquad \JfiveTET = \tr\left(\bMtet:\bC^2\right)^2\,, 
\end{aligned}
\end{equation}
form a minimal integrity basis for the tetragonal $\cD_4$ group.

\paragraph{(B) Polyconvex integrity basis} Polyconvex invariants for the tetragonal symmetry group $\cD_{4}$ were first proposed by \textcite[Sec.~3]{Schroeder2010a}. However, therein, no integrity or functional basis is provided. We now propose a new polyconvex integrity basis for the tetragonal $\cD_4$ group. To achieve this, we make use of a new method for invariant construction that is based on a group symmetrization approach (cf.~Theorem~\ref{th:group_symm}):

\begin{theorem}[\textbf{Polyconvex tetragonal invariant function based on group symmetrization}]\label{theorem:new_tetragonal_invs}
The function
\begin{equation}\label{eq:new_tetragonal_invs}
    \begin{aligned}
f=\sum_{i=1}^2\left(a_1\left[\bn_{0}^{(i)}\otimes\bn_{0}^{(i)}:\bC\right]^{b_1}+a_2\left[ \bn_{0}^{(i)}\otimes\bn_{0}^{(i)}:\bG\right]^{b_2}\right)^{c}\quad \text{with}\quad a_1,\,a_2\geq 0,\, b_1,\,b_2,c\geq 1\,,
    \end{aligned}
\end{equation}
is a polyconvex invariant of the tetragonal symmetry group $\cD_4$.%
\end{theorem}
\begin{proof}
Let us consider the polyconvex function
\begin{equation}
\widetilde{f}=2\left(a_1\left[\bn_{0}^{(1)}\otimes\bn_{0}^{(1)}:\bC\right]^{b_1}+a_2\left[\bn_{0}^{(1)}\otimes\bn_{0}^{(1)}:\bG\right]^{b_2}\right)^{c}\,.
\end{equation}
Application of the group symmetrization for the tetragonal $\cD_4$ group on $\widetilde{f}$ yields
\begin{equation}\label{eq:ifour_GS_tet}
\begin{aligned}
    f&=\frac{1}{8}\sum_{\bQ\in\cD_4}2\left(a_1\left[\bn_{0}^{(1)}\otimes\bn_{0}^{(1)}:\left(\bQ\cdot\bC\cdot\bQ^T\right)\right]^{b_1}+a_2\left[\bn_{0}^{(1)}\otimes\bn_{0}^{(1)}:\left(\bQ\cdot\bG\cdot\bQ^T\right)\right]^{b_2}\right)^{c}
    \\
   & =\frac{1}{8}\sum_{\bQ\in\cD_4}2\left(a_1\left[\left(\bQ^T\cdot\bn_{0}^{(1)}\otimes\bQ^T\cdot\bn_{0}^{(1)}\right):\bC\right]^{b_1}+a_2\left[\left(\bQ^T\cdot\bn_{0}^{(1)}\otimes\bQ^T\cdot\bn_{0}^{(1)}\right):\bG\right]^{b_2}\right)^{c}
   \\
&=\sum_{i=1}^2\left(a_1\left[\bn_{0}^{(i)}\otimes\bn_{0}^{(i)}:\bC\right]^{b_1}+a_2\left[ \bn_{0}^{(i)}\otimes\bn_{0}^{(i)}:\bG\right]^{b_2}\right)^{c}\,.
\end{aligned}
\end{equation}
In~\cref{eq:ifour_GS_tet}, the last equality follows by evaluating $\bQ^T\cdot\bn_{0}^{(1)}$ for the 8 elements of the tetragonal symmetry group \parencite[Table~1]{spencer1971}. {Since this symmetry group has a finite number of elements and $\widetilde{f}$ is polyconvex, Theorem~\ref{th:group_symm} states that the function $f$ in~\cref{eq:new_tetragonal_invs} is a polyconvex invariant of the tetragonal $\cD_4$ group.}
\end{proof}
\begin{corollary}[\textbf{Polyconvex tetragonal invariants based on group symmetrization}]\label{corr:new_tetragonal_invs}
Based on~Theorem~\ref{theorem:new_tetragonal_invs}, we construct the polyconvex invariants\footnote{We would like to point out that the polyconvex bases determined by group symmetrization, unlike the integrity bases known from the literature, contain some inhomogeneous polynomials.}
\begin{equation}\label{eq:tet_invs_pc_basis}
\begin{aligned}
     &  \IoneTET= \sum_{i=1}^2\bn_{0}^{(i)}\otimes\bn_{0}^{(i)}:\bC\,,
\qquad\ItwoTET=\sum_{i=1}^2\left[\bn_{0}^{(i)}\otimes\bn_{0}^{(i)}:\bC\right]^2\,,
\qquad \IthreeTET =\sum_{i=1}^2\bn_{0}^{(i)}\otimes\bn_{0}^{(i)}:\bG\,,
\\
   & \IfourTET=\sum_{i=1}^2\left[\bn_{0}^{(i)}\otimes\bn_{0}^{(i)}:\bG\right]^2\,,
\qquad\IfiveTET=\sum_{i=1}^2\left(\left[\bn_{0}^{(i)}\otimes\bn_{0}^{(i)}:\bC\right]+\left[\bn_{0}^{(i)}\otimes\bn_{0}^{(i)}:\bG\right]\right)^2\,.
\end{aligned}
\end{equation}
The invariants $I^{\text{tet}}_{1-5}$ and $I^{\text{iso}}_{1-3}$ form an integrity basis for the tetragonal $\cD_4$ group. Following \textcite[Prop.~1]{riemer2025}, this is because $I^{\text{tet}}_{1-5}$ and $I^{\text{iso}}_{1,2}$ can represent $J^{\text{iso}}_{1,2}$ and $J^{\text{tet}}_{1-5}$ as polynomials, cf.~App.~\ref{app:inv_relations_iso} and~App.~\ref{app:inv_relations_tet}. Furthermore, the representation of $I^{\text{tet}}_{1-5}$ in terms of $J^{\text{tet}}_{1-5}$ in~App.~\ref{app:inv_relations_tet} verifies that the former are invariants of the tetragonal $\cD_4$ group.
\end{corollary}

\paragraph{(C) Extension of the polyconvex integrity basis} Based on~\cref{eq:new_tetragonal_invs}, additional polyconvex tetragonal invariants could be constructed that might improve the flexibility of the constitutive model. Moreover, \cref{eq:new_tetragonal_invs} could be employed as part of the constitutive model with trainable material parameters $\{a_1,a_2,b_1,b_2,c\}$.


\subsubsection{Cubic material behavior}\label{sec:invs_CUB}

Two cubic symmetry groups are relevant in anisotropic hyperelasticity \parencite{riemer2025}. In this work, we only consider the cubic $\cO$ group.

\paragraph{(A) Minimal integrity basis}

For the cubic symmetry group $\cO$, a single fourth-order structural tensor is sufficient to characterize the material's anisotropy, e.g.,
\begin{equation}\label{eq:cub_struct}
   \cS^{\text{cub}}=\big(\bMcub\big)\qquad\text{with}\qquad \bMcub=\sum_{i=1}^3 \bn_{0}^{(i)}\otimes\bn_{0}^{(i)}\otimes\bn_{0}^{(i)}\otimes\bn_{0}^{(i)}\,.
\end{equation}
Together with the isotropic invariants $J^{\text{iso}}_{1-3}$ introduced in~\cref{sec:invs_iso}, the invariants
\begin{equation}\label{eq:invs_cub_karl}
\begin{aligned}
&{\mathrlap{\JoneCUB=\tr\left(\bMcub:\bC\right)^2\,,}\phantom{\JfourCUB=\tr\left(\bMcub:\bC^2\right)^2\,,}}\qquad {\mathrlap{ \JtwoCUB=\tr\left(\bMcub:\bC\right)^3\,,}\phantom{\JfiveCUB=\tr\bC^2\cdot\left(\bMcub:\bC\right)^2\,,}}\qquad
{\mathrlap{\JthreeCUB=\tr\bC^2\cdot\left(\bMcub:\bC\right)\,, }\phantom{\JsixCUB=\tr\bC\cdot\left(\bMcub:\bC^2\right)^2\,,}}
\\
&{\mathrlap{\JfourCUB=\tr\left(\bMcub:\bC^2\right)^2\,,}\phantom{\JfourCUB=\tr\left(\bMcub:\bC^2\right)^2\,,}}\qquad{\mathrlap{\JfiveCUB=\tr\bC^2\cdot\left(\bMcub:\bC\right)^2\,,}\phantom{\JfiveCUB=\tr\bC^2\cdot\left(\bMcub:\bC\right)^2\,,}}\qquad
{\mathrlap{\JsixCUB=\tr\bC\cdot\left(\bMcub:\bC^2\right)^2\,,}\phantom{\JsixCUB=\tr\bC\cdot\left(\bMcub:\bC^2\right)^2\,,}}
\end{aligned}
\end{equation}
form a minimal integrity basis for the cubic $\cO$ group.

\paragraph{(B) Polyconvex functional basis} Polyconvex invariants for the cubic symmetry group $\cO$ were first proposed by \textcite[Sec.~3]{Kambouchev2007}. However, therein, no integrity or functional basis is provided. We now propose a new polyconvex functional basis for the cubic $\cO$ group. To achieve this, we make again use of the new method for invariant construction that is based on a group symmetrization approach (cf.~Theorem~\ref{th:group_symm}):
\begin{theorem}[\textbf{Polyconvex cubic invariant function based on group symmetrization}]\label{theorem:new_cub_invs}
The function
\begin{equation}\label{eq:new_cubic_invs}
    \begin{aligned}
f=\sum_{i=1}^3\left(a_1\left[\bn_{0}^{(i)}\otimes\bn_{0}^{(i)}:\bC\right]^{b_1}+a_2\left[ \bn_{0}^{(i)}\otimes\bn_{0}^{(i)}:\bG\right]^{b_2}\right)^{c}\quad \text{with}\quad a_1,\,a_2\geq 0,\, b_1,\,b_2,c\geq 1\,,
    \end{aligned}
\end{equation}
is a polyconvex invariant of the cubic symmetry group $\cO$.
\end{theorem}
\begin{proof}
Let us consider the polyconvex function
\begin{equation}
\widetilde{f}=3\left(a_1\left[\bn_{0}^{(1)}\otimes\bn_{0}^{(1)}:\bC\right]^{b_1}+a_2\left[\bn_{0}^{(1)}\otimes\bn_{0}^{(1)}:\bG\right]^{b_2}\right)^{c}\,.
\end{equation}
Application of the group symmetrization for the cubic $\cO$ group on $\widetilde{f}$ provides
\begin{equation}\label{eq:ifour_GS}
\begin{aligned}
    f&=\frac{1}{24}\sum_{\bQ\in\cO}3\left(a_1\left[\bn_{0}^{(1)}\otimes\bn_{0}^{(1)}:\left(\bQ\cdot\bC\cdot\bQ^T\right)\right]^{b_1}+a_2\left[\bn_{0}^{(1)}\otimes\bn_{0}^{(1)}:\left(\bQ\cdot\bG\cdot\bQ^T\right)\right]^{b_2}\right)^{c}
    \\
   & =\frac{1}{24}\sum_{\bQ\in\cO}3\left(a_1\left[\left(\bQ^T\cdot\bn_{0}^{(1)}\otimes\bQ^T\cdot\bn_{0}^{(1)}\right):\bC\right]^{b_1}+a_2\left[\left(\bQ^T\cdot\bn_{0}^{(1)}\otimes\bQ^T\cdot\bn_{0}^{(1)}\right):\bG\right]^{b_2}\right)^{c}
   \\
&=\sum_{i=1}^3\left(a_1\left[\bn_{0}^{(i)}\otimes\bn_{0}^{(i)}:\bC\right]^{b_1}+a_2\left[ \bn_{0}^{(i)}\otimes\bn_{0}^{(i)}:\bG\right]^{b_2}\right)^{c}\,.
\end{aligned}
\end{equation}
In~\cref{eq:ifour_GS}, the last equality follows by evaluating $\bQ^T\cdot\bn_{0}^{(1)}$ for the 24 elements of the cubic symmetry group \parencite[Table~1]{spencer1971}. {Since this symmetry group has a finite number of elements and $\widetilde{f}$ is polyconvex, Theorem~\ref{th:group_symm} states that the function $f$ in~\cref{eq:new_cubic_invs} is a polyconvex invariant of the cubic $\cO$ group.}
\end{proof}
\begin{corollary}[\textbf{Polyconvex cubic invariants based on group symmetrization}]\label{corr:new_cubic_invs}
Based on~Theorem~\ref{theorem:new_cub_invs}, we construct the polyconvex invariants
\begin{equation}\label{eq:cub_invs_pc_basis}
\begin{aligned}
     &  \IoneCUB= \sum_{i=1}^3\left[\bn_{0}^{(i)}\otimes\bn_{0}^{(i)}:\bC\right]^2\,,
&&\ItwoCUB=\sum_{i=1}^3\left[\bn_{0}^{(i)}\otimes\bn_{0}^{(i)}:\bC\right]^3\,,
\\
& \IthreeCUB =\sum_{i=1}^3\left[\bn_{0}^{(i)}\otimes\bn_{0}^{(i)}:\bC+\bn_{0}^{(i)}\otimes\bn_{0}^{(i)}:\bG\right]^2\,,
   && \IfourCUB=\sum_{i=1}^3\left[\bn_{0}^{(i)}\otimes\bn_{0}^{(i)}:\bG\right]^2\,,
\\
&\IfiveCUB=\sum_{i=1}^3\left(\left[\bn_{0}^{(i)}\otimes\bn_{0}^{(i)}:\bC\right]^2+\left[\bn_{0}^{(i)}\otimes\bn_{0}^{(i)}:\bG\right]\right)^2\,,
  &&\IsixCUB=\sum_{i=1}^3\left[\bn_{0}^{(i)}\otimes\bn_{0}^{(i)}:\bG\right]^3\,.
    %
\end{aligned}
\end{equation}
The invariants $I^{\text{cub}}_{1-6}$ and $I^{\text{iso}}_{1-3}$ can represent $J^{\text{cub}}_{1-6}$ and $J^{\text{iso}}_{1-3}$, see App.~\ref{app:inv_relations_iso} and \ref{app:inv_relations_cub}. Thereby, $I^{\text{cub}}_{1-5}$ can be formulated as polynomials of $I^{\text{cub}}_{1-5}$ and $I^{\text{iso}}_{1-3}$. However, the representation of $\JsixCUB$ in terms of $I^{\text{cub}}_{1-6}$ and $I^{\text{iso}}_{1-3}$ is not a polynomial, as it includes a division by $\IoneISO$, cf.~\cref{eq:Jsix_representation}. Thus, $I^{\text{cub}}_{1-6}$ and $I^{\text{iso}}_{1-3}$ form only a functional basis but not an integrity basis for cubic anisotropy. Furthermore, the representation of $I^{\text{cub}}_{1-6}$ in terms of $J^{\text{cub}}_{1-6}$ in~App.~\ref{app:inv_relations_cub} verifies that the former are invariants of the cubic $\cO$ group.
\end{corollary}

\paragraph{(C) Extension of the polyconvex integrity basis} Based on~\cref{eq:new_cubic_invs}, additional polyconvex cubic invariants could be constructed that might improve the flexibility of the constitutive model. Moreover, \cref{eq:new_cubic_invs} could be employed as part of the constitutive model with trainable material parameters $\{a_1,a_2,b_1,b_2,c\}$.

{
\begin{remark}
    Recall that group symmetrization ensures exact fulfillment of the material symmetry condition only for \emph{finite} symmetry groups (cf.~Remark~\ref{rem:finite_infinite}). Consequently, the construction of polyconvex invariants via group symmetrization cannot be employed for \emph{infinite} symmetry groups. However, for all relevant infinite symmetry groups, polyconvex integrity bases constructed via structural tensors have already been established, which we review in~Appendix\ref{app:invs_established_struct}.
\end{remark}
}


\section{Polyconvex PANN constitutive modeling}\label{sec:PANNs}

In this section, we discuss different polyconvex constitutive models based on physics-augmented neural networks (PANNs). The PANN constitutive models are based on monotonic and convex neural networks which are introduced in~\cref{sec:CMNNs}. Subsequently, in~\cref{sec:PANN_comp,sec:PANN_C}, we discuss polyconvex PANN modeling approaches based on invariants constructed via structural tensors,\footnote{{For this modeling approach, coordinate-based invariants or invariants constructed via group symmetrization can also be applied, cf.~Remark~\ref{rem:diff_inv_deriv}.}} as well as an approach based on triclinic invariants and group symmetrization. 

{To recall Remark~\ref{rem:diff_inv_deriv}, the PANN models introduced in the following mainly differ in how they incorporate the material symmetry condition. The model based on invariants constructed via structural tensors directly encodes the material (an-)isotropy in the chosen tuple of invariants. In contrast, the model based on triclinic invariants and group symmetrization is initially formulated in terms of triclinic invariants, which represents the most general form of anisotropy, while the desired material symmetry is included only in a subsequent group-symmetrization step rather than through the invariant set itself.}

\subsection{Monotonic and convex neural networks}\label{sec:CMNNs}

In this work, we consider constitutive models based on feed-forward neural networks (FFNNs), which are a common type of artificial neural network \parencite{herrmann2025,aggarwal2018}. In a nutshell, FFNNs consist of multiple compositions of vector-valued layers, where each layer alternates between vector-valued affine transformations and element-wise applied nonlinear functions. The components in the layers are referred to as nodes or neurons, and the function acting in each node is called the activation function. The number of nodes in a layer is referred to as width and the number of layers is referred to as depth of the NN. We consider scalar-valued FFNNs given as the mapping $\cF:\bbR^{n}\rightarrow \bbR\,,\,\, \bx\mapsto \cF(\bx)$, where $\cF$ is recursively defined as
\begin{equation}\label{eq:FFNN_1}
\begin{aligned}
{\mathrlap{\bx^{(h)}}\phantom{\bx^{(h)}}}&=\SP\big(\bW^{(h)}\cdot\bx^{(h-1)}+\bb^{(h)}\big)&&\in\bbR^{n_h}\,,\quad h=1,\,\dotsc H \,,
\\
{\mathrlap{\cF}\phantom{\bx^{(h)}}}&=\bw^{(H+1)}\cdot\bx^{(H)}&&\in\bbR\,,
\end{aligned}
\end{equation}
with $\bx^{(0)}=\bx$ and $n_0=n$. The layers from 1 to $H$ are commonly referred to as hidden layers, while the layer $H+1$ is referred to as output layer.
The weight matrices $\bW^{(h)}\in\bbR^{n_{h}\times n_{h-1}},\bw^{(H+1)}\in\bbR^{n_{H}}$ and the bias vectors $\bb^{(h)}\in\bbR^{n_h}$ form the set of parameters $\boldsymbol{\mathcal{P}}$ that are optimized to fit the NN to a given dataset. {As the width or depth of the NN increases, both its number of trainable parameters and its flexibility grow. Throughout this work, we employ the softplus activation function $\SP(x)=\log(1+e^x)$ in the hidden layers, which is applied component-wise. In the output layer, we employ a linear activation function and no bias.}  

As originally proposed by \textcite{Amos2017}, the simple structure and recursive definition of \cref{eq:FFNN_1} makes FFNNs a very natural choice for the construction of monotonic and convex functions. {As all employed activation functions are convex and monotonic, we may restrict~\cref{eq:FFNN_1} as follows. If the weights in every layer except the first are non-negative, i.e., $(\bw^{(H+1)})_{i},\,(\bW^{(h)})_{ij}\geq 0$ for $h>1$, $\cF$ is convex in $\bx$. We call this an \textbf{input-convex neural network (ICNN)}. Moreover, if all weights are non-negative, $\cF$ is convex and monotonic in $\bx$. We call this a \textbf{convex-monotonic neural network (CMNN)}.} For an explicit proof of convexity and monotonicity, the reader is referred to \textcite[Appendix~A]{klein2022a}. 

\subsection{Polyconvex PANN model based on invariants constructed via structural tensors}\label{sec:PANN_comp}

In the following, we briefly introduce the polyconvex PANN framework proposed by~\textcite{klein2022a,linden2023}. Thereby, we consider FFNNs that take the invariants $\bcI^{\square}\in\bbR^m$ as input, and have a scalar-valued output that is used as a part of the hyperelastic potential:
\begin{equation}\label{eq:pot_NN}
\psi^{\text{NN}}_{\square}\big(\bcI^{\square}\big)=\cF\big(\bcI^{\square}\big)\,,
\end{equation}
with the neural network $\cF$ as defined in~\cref{eq:FFNN_1}. Here, $\square$ denotes the considered symmetry group for which the model is formulated. The choice of invariants $\bcI^{\square}$  depends on the symmetry group being examined and whether the model is constructed as polyconvex, where different sets of invariants are introduced in \cref{sec:inv_bases_group_symm,app:invs_established_struct}. As originally proposed by~\textcite{klein2022a}, a polyconvex NN potential can be constructed when all invariants in $\bcI^{\square}$ are polyconvex and $\psi^{\text{NN}}_{\square}$ is represented by a CMNN, cf.~\cref{sec:inv_conv,sec:CMNNs}. Note that to this point, the NN potential only fulfills thermodynamic consistency, objectivity, the material symmetry condition, {and potentially polyconvexity} by construction.
 
To fulfill \emph{all} common constitutive conditions of hyperelasticity, the NN potential is complemented by additional growth and normalization terms. This results in the overall PANN model given by
\begin{equation}\label{eq:PANN}
    \psi^{\text{PANN}}_{\square}(\bcI^{\square})=\psi^{\text{NN}}_{\square}(\bcI^{\square})+\psi^{\text{stress}}_{\square}(\bcI^{\square})+\psi^{\text{growth}}(J)\,.
\end{equation}
Here, $\psi^{\text{stress}}_{\square}$ and $\psi^{\text{growth}}$ denote the normalization and growth terms, which ensure fulfillment of the stress normalization condition and the volumetric growth condition (cf.~\cref{sec:const_cond}). The overall flow and structure of the PANN model is visualized in~\cref{fig:PANNs:I}. The first Piola-Kirchhoff stress of the PANN model can be calculated via~\cref{eq:PK1}. Following \textcite{klein2022a}, the growth term $\psi^\text{growth}$ is chosen as 
\begin{equation}\label{eq:PANN_growth}
    \psi^\text{growth}(J):=\alpha\left(J+J^{-1}-2\right)^2\,,\quad \alpha>0\,,
\end{equation}
which is convex in $J$ and thus polyconvex \parencite{Hartmann2003}. Since $\psi^\text{growth}$ only depends on the isotropic invariant $J$, it is compatible with every symmetry group. 

Furthermore, we employ the stress normalization approach proposed by \textcite{linden2023}, which is compatible with all discussed constitutive conditions of compressible hyperelasticity. In~\textcite{linden2023}, polyconvex normalization terms are introduced for the isotropic $\cK$ and the transversely isotropic $\cD_{\infty}$ groups. We recall the normalization approach for transverse isotropy and exemplary extend the approach for cubic anisotropy. Note that corresponding terms could be formulated for the other symmetry groups.

\paragraph{Transversely isotropic normalization term}

Let us consider a transversely isotropic potential $\psi^{\text{NN}}_{\text{ti}}$ that depends on the invariants $(\IoneISO,\,\ItwoISO,\,\IthreeISO,\,\IfourISO,\,\IoneTI,\,\ItwoTI)$ introduced in~\cref{sec:invs_TI}.\footnote{{If the potential also depends on the additional invariants $\IthreeTI$ and $\IfourTI$ introduced in~\cref{eq:invs_ti_2}, the stress normalization term must be slightly adapted, cf.~\parencite[Sec.~4.4.3]{kleinDiss}.}} In this case, the stress normalization term
\begin{equation}\label{eq:stress_norm_comp_TI}
    \psi^{\text{stress}}_{\text{ti}}(J,\,\IoneTI,\,\ItwoTI)=-2(\fo+\fq)\,J+ \fp\,\IoneTI+ \fq\,\ItwoTI\quad \text{with}\quad \fp={\RL(-x)}\,,\,\,
 \fq={\RL(x)}\in\bbR^+\,,
\end{equation}
proposed by \textcite{linden2023} can be applied, with the constants
\begin{equation}\label{eq:stress_norm_TI_coeff}
\begin{aligned}
    \fo&=\Bigg[
\partial_{\IoneISO}\psi^{\text{NN}}_{\text{ti}}+2\partial_{\ItwoISO}\psi^{\text{NN}}_{\text{ti}}+\frac{1}{2}\left(\partial_{\IthreeISO}\psi^{\text{NN}}_{\text{ti}}-\partial_{\IfourISO}\psi^{\text{NN}}_{\text{ti}}\right)
+\partial_{\ItwoTI}\psi^{\text{NN}}_{\text{ti}}
\Bigg]\Bigg\rvert_{\bF=\bI}\in\bbR\,,
    \end{aligned}
\end{equation}
and
\begin{equation}
    x=\Bigg[\partial_{\IoneTI}\psi^{\text{NN}}_{\text{ti}} - \partial_{\IthreeTI}\psi^{\text{NN}}_{\text{ti}} - \partial_{\ItwoTI}\psi^{\text{NN}}_{\text{ti}} + \partial_{\IfourTI}\psi^{\text{NN}}_{\text{ti}}    \Bigg]\Bigg\rvert_{\bF=\bI}
    \in\bbR\,,
\end{equation}
where $\RL(x)=\max(x,\,0)$. This stress normalization term is polyconvex as it is linear in $J$ and the constants $\fp,\,\fq$ in~\cref{eq:stress_norm_comp_TI} are non-negative. After {training}, $\fp,\,\fq$ remain constant, and the non-differentiability of the $\RL$ function in~\cref{eq:stress_norm_comp_TI} poses no problem.


\paragraph{Cubic normalization term}

Based on the method proposed by \parencite{linden2023}, we introduce the new cubic stress normalization term
\begin{equation}\label{eq:stress_norm_comp_CUB}
    \psi^{\text{stress}}_{\text{cub}}(J,\,\IoneCUB,\,\ItwoCUB)=-2\left(\fr+ 3\ft\right)\,J+ \frac{1}{2}\fs\,\IoneCUB+ \frac{1}{2}\ft\,\IfourCUB\quad \text{with}\quad
 \fs={\RL(-x)}\,,\,\,
  \ft={\RL(x)}\in\bbR^+\,,
\end{equation}
with the constants
\begin{equation}\label{eq:stress_norm_CUB_coeff}
\begin{aligned}
    \fr&=\Bigg[
\partial_{\IoneISO}\psi^{\text{NN}}_{\text{cub}}+2\partial_{\ItwoISO}\psi^{\text{NN}}_{\text{cub}}+\frac{1}{2}\left(\partial_{\IthreeISO}\psi^{\text{NN}}_{\text{cub}}-\partial_{\IfourISO}\psi^{\text{NN}}_{\text{cub}}\right)
\\
&\phantom{=2\Bigg[}
+6\partial_{\IfourCUB}\psi^{\text{NN}}_{\text{cub}}+12\left(\partial_{\IthreeCUB}\psi^{\text{NN}}_{\text{cub}}+ \partial_{\IfiveCUB}\psi^{\text{NN}}_{\text{cub}}\right)+9\partial_{\IsixCUB}\psi^{\text{NN}}_{\text{cub}}
\Bigg]\Bigg\rvert_{\bF=\bI}\in\bbR\,,
    \end{aligned}
\end{equation}
and
\begin{equation}
\begin{aligned}
x&=\Bigg[2\left(\partial_{\IoneCUB}\psi^{\text{NN}}_{\text{cub}}-\partial_{\IfourCUB}\psi^{\text{NN}}_{\text{cub}}
    \right)    
    +3\left(\partial_{\ItwoCUB}\psi^{\text{NN}}_{\text{cub}}-\partial_{\IsixCUB}\psi^{\text{NN}}_{\text{cub}} \right)
+4 \partial_{\IfiveCUB}\psi^{\text{NN}}_{\text{cub}}\Bigg]\Bigg\rvert_{\bF=\bI}
    \in\bbR\,.
\end{aligned}
\end{equation}
This stress normalization term is polyconvex as it is linear in $J$ and the constants $\fs,\,\ft$ in~\cref{eq:stress_norm_comp_CUB} are non-negative. Note that~\cref{eq:stress_norm_comp_CUB} is only one possible choice for the construction of polyconvex stress normalization terms for cubic anisotropy, and corresponding approaches could be formulated making use of different cubic invariants. However, in preliminary studies, the choice of the cubic stress normalization term had no influence on the performance of the PANN model. 

 \begin{figure}
         \centering
    \begin{subfigure}[t]{.49\textwidth}
        \centering
        \includegraphics[width=\textwidth]{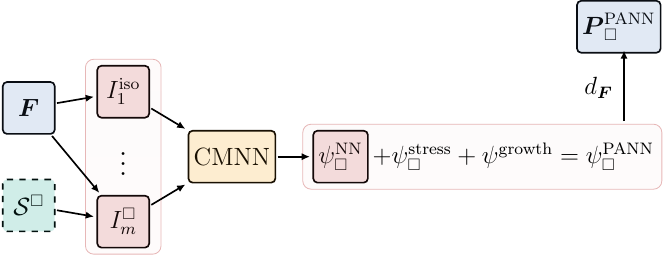}
         \caption{PANN--$\cI$ model using structural tensor-based invariants.}
         \label{fig:PANNs:I}
    \end{subfigure}%
    \hfill
    \begin{subfigure}[t]{.49\textwidth}
        \centering
    \includegraphics[width=\textwidth]{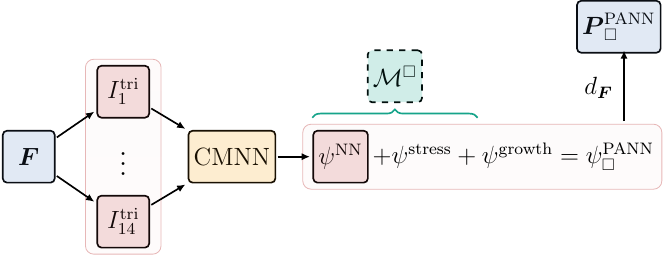}
         \caption{PANN--$\cC$ model based on triclinic invariants and group symmetrization.}     
         \label{fig:PANNs:C}
    \end{subfigure}%
    \caption{Illustration of the different polyconvex PANN constitutive models. The models include material symmetry by either structural tensors $\cS^{\square}$ or a group symmetrization operation $\cM^{\square}$, where $\square$ denotes the considered symmetry group.}
    \label{fig:PANNs}
 \end{figure}
 
\subsection{Polyconvex PANN model based on triclinic invariants and group symmetrization}\label{sec:PANN_C}

{We introduced a novel polyconvex constitutive modeling framework based on triclinic invariants and group symmetrization in~\cref{sec:aniso_hyperel_tric_group}. By construction, it satisfies thermodynamic consistency, material symmetry (only for finite symmetry groups, cf.~Remark~\ref{rem:finite_infinite}), objectivity, and polyconvexity. The remaining constitutive requirements, namely stress normalization and the volumetric growth condition, require additional modifications that may differ between specific model formulations. In the following, we assume that the considered symmetry group is finite, and exemplarily present a PANN model that also fulfills these remaining requirements.
 
The framework starts with the convex and monotonic potential~\cref{eq:pot_c_0} which we now parameterize with a NN, i.e.,
\begin{equation}\label{eq:pot_NN}
\psi^{\text{NN}}\big(\bcI^{\text{tri}}\big)=\cF\big(\bcI^{\text{tri}}\big)\,,
\end{equation}
where $\bcI^{\text{tri}}\in\bbR^{14}$ is the triclinic polyconvex integrity basis defined in~\cref{eq:K_basis_recall}, and where $\cF$ is constrained to be a CMNN, cf.~\cref{sec:CMNNs}. To ensure stress normalization, the potential $\psi^{\text{NN}}$ is complemented by the normalization term $\psi^{\text{stress}}$, i.e.,
\begin{equation}
\widetilde{\psi}^{\text{NN}}(\bcI^{\text{tri}})=\psi^{\text{NN}}( \bcI^{\text{tri}})+\psi^{\text{stress}}( \bcI^{\text{tri}})\,.
\end{equation}
For the construction of the stress normalization term, we employ the fact that the projections used to construct in $\bcI^{\text{tri}}$ (cf.~\cref{eq:C_comp}) have the same structure as the transversely isotropic invariants introduced in~\cref{sec:invs_TI}. Thus, the polyconvex invariant-based stress normalization terms proposed by \textcite{linden2023} for transverse isotropy can immediately be employed (cf.~\cref{eq:stress_norm_comp_TI}). For each projection direction, one such normalization term can be formulated, and the overall stress normalization term $\psi^{\text{stress}}$ is their sum. Then, by complementing the potential by a growth term $\psi^\text{growth}(J)$, the volumetric growth condition can be fulfilled. Here, we employ the growth term introduced in~\cref{eq:PANN_growth}. The overall PANN constitutive model is given by
\begin{equation}\label{eq:PANN_C}
\psi^{\text{PANN}}_{\square}=\frac{1}{\lvert\mathcal{G}\rvert}\sum_{\bQ\in\mathcal{G}}\Big[\widetilde{\psi}^{\text{NN}}\big( \bcI^{\text{tri}}(\bQ\star\bC)\big)\Big]+ \psi^\text{growth}(J)\,.
\end{equation}
where we employ the group symmetrization discussed in~\cref{sec:group_symm_intro}. To recall its most important properties, for finite symmetry groups, group symmetrization ensures exact fulfillment of the material symmetry condition. Furthermore, it preserves polyconvexity of a function. Thus, for finite symmetry groups, the PANN model introduced in~\cref{eq:PANN_C} fulfills all relevant constitutive conditions of hyperelasticity by construction. 
The overall flow and structure of the model is visualized in~\cref{fig:PANNs:C}. The first Piola-Kirchhoff stress of the PANN model can be calculated via~\cref{eq:PK1}.}
{
\begin{remark}
The increased computational cost of this PANN model, cf.~Remark~\ref{rem:comput_cost}, can be partly mitigated with an efficient batch-wise evaluation of the NN potential \parencite{aggarwal2018,ALHEIT2026118728}. For instance, throughout the numerical examples considered in~\cref{sec:application_micro_calib}, the computational time required for training and evaluation of this model was approximately seven times higher than that of the polyconvex PANN model based on invariants constructed via structural tensors. Note that the computation time depends on the specific implementation and available computer hardware. Moreover, instead of fulfilling material symmetry in an exact fashion through group symmetrization, it can be approximately fulfilled by augmenting the {training} data with the known symmetry relations \parencite{gaertner2021}. Furthermore, this model could also be employed without incorporating material symmetry at all, which could be beneficial in scenarios where no information about the underlying symmetry of the material is available.    
\end{remark}
}


\section[Application to homogenization data of microstructured materials]{Application to homogenization data of cubic metamaterials}\label{chap:application_micro}

In this section, we apply different polyconvex PANN constitutive models to synthetic homogenization data of cubic metamaterials. In~\cref{sec:data_models}, we provide information on the considered datasets, PANN models, and {model training} strategy. This is followed by an evaluation of the {trained} PANN models in~\cref{sec:evaluation}. 

\subsection{Data and PANN model preparation}\label{sec:data_models}

We now provide technical details on the considered datasets and PANN constitutive models.

\subsubsection{Considered homogenization data}\label{sec:data_details}

\begin{figure}[t]
	\centering
    \begin{subfigure}[t]{.3\textwidth}
        \centering
         \includegraphics[height=2.8cm]{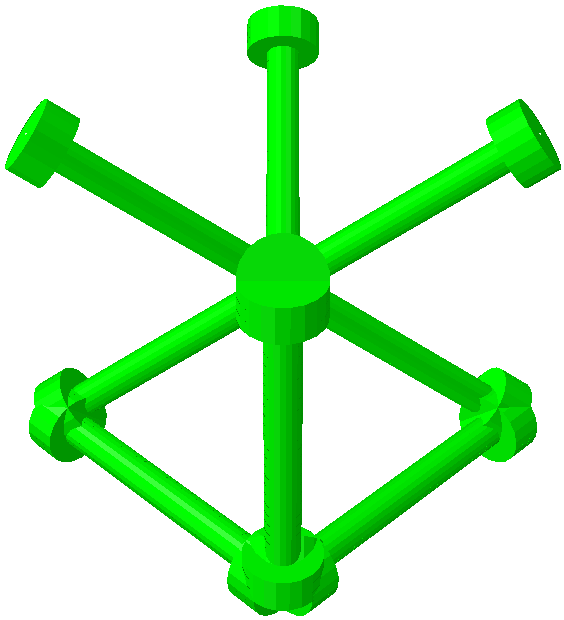}
         \caption{BCC lattice metamaterial.\\ Image from \parencite{Fernandez2020}.}
         \label{fig:microstructures:BCC}
    \end{subfigure}%
    ~~~~~
    \begin{subfigure}[t]{.35\textwidth}
        \centering
         \includegraphics[height=2.8cm]{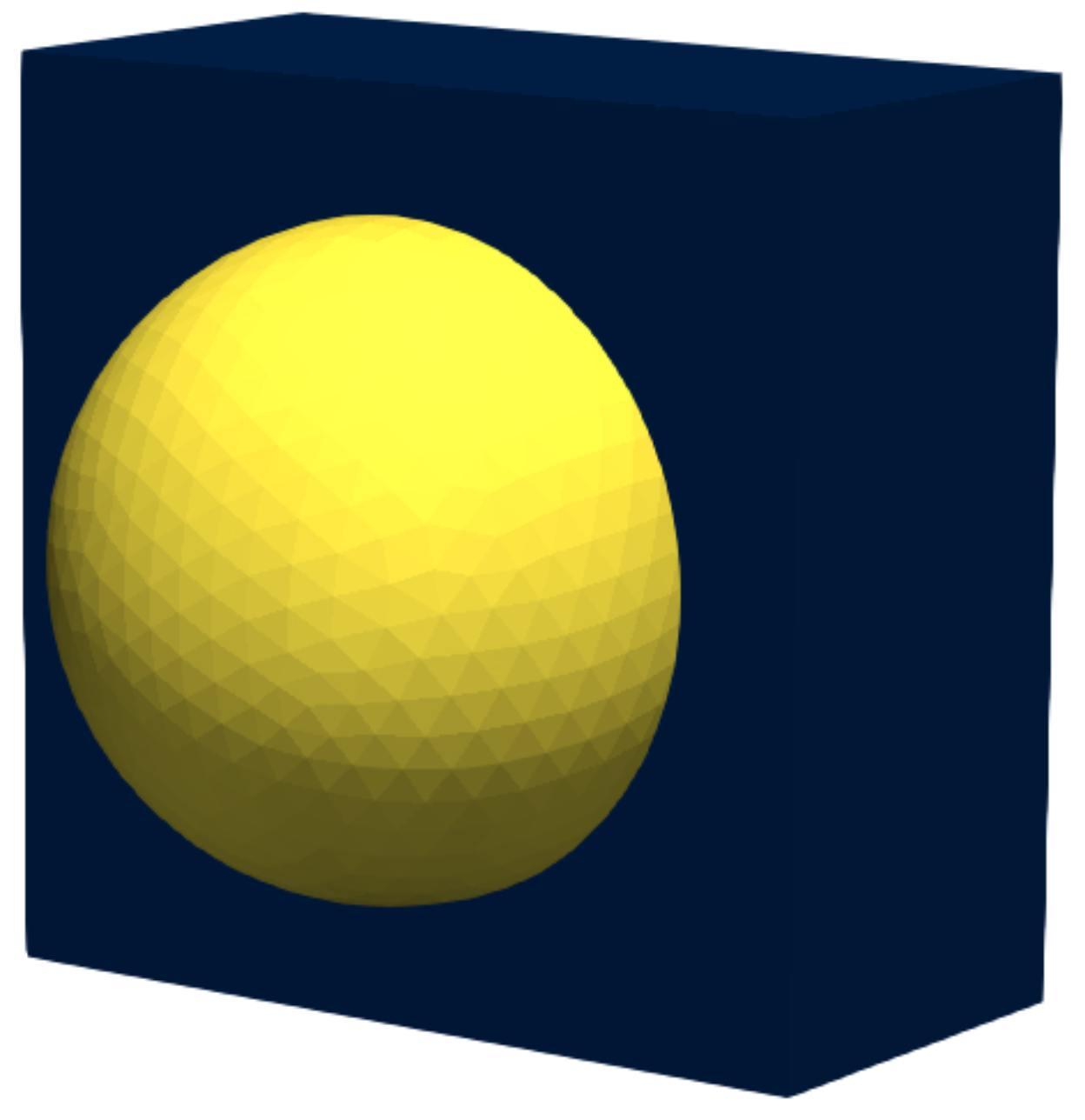}
         \caption{Matrix with spherical inclusion (SPH).\\ Image from \parencite{kalina2025}.}
         \label{fig:microstructures:SPH}
    \end{subfigure}%
	\caption{Microstructures considered in this work.
	}
	\label{fig:microstructures}
\end{figure}

We consider synthetic homogenization data of two different materials with heterogeneous microstructures with cubic material symmetry, which are visualized in~\cref{fig:microstructures}. In both cases, the effective material behavior can be described with the cubic $\cO$ group \parencite{Fernandez2020,kalina2025}. The considered metamaterials are briefly introduced in the following. 

\paragraph{Body-centered cubic lattice metamaterial (BCC)} A BCC structure consists of  body-centered beams with additional beams along all edges, see \cref{fig:microstructures:BCC}. The smallest repeating unit, referred to as the BCC unit cell, contains beams only at three edges due to the periodicity of the structure. The effective material behavior has a cubic anisotropy. We consider data from \textcite{Fernandez2020}, where a BCC cell is numerically homogenized. In this study, the beams are modeled as Timoshenko beam elements composed of a linear elastic, isotropic material. Nevertheless, the effective material behavior of the BCC unit cell is highly nonlinear due to geometrically nonlinear deformations.

For the data generation, load scenarios inspired by physical experiments are applied, such as uniaxial deformation and shear. To evaluate the model performance under general deformation states, more complex load scenarios such as a mixed shear-tension test (``test 3'' in \textcite{Fernandez2020}) are also considered. We only consider a part of the overall dataset provided by \textcite{Fernandez2020}, in particular, the uniaxial load case, the shear case, and the mixed shear-tension test. The different load paths have 175, 101, and 172 datapoints, respectively. We employ uniaxial tension and shear for {training} and the mixed case for testing.
In the invariant space, the considered load scenarios are within $\IoneISO\in[2.5,3.3]$ and $\ItwoISO\in[1.9,3.3]$, thus covering a finite deformation regime. From the homogenization, stress--strain datasets of the form
\begin{equation}\label{eq:dataset_mech_2}
        \cD=\Big\{\left({^1\bF},\, {^1\bP}\right),\dotsc,
        \left({^m\bF},\,{^m\bP}\right)\Big\}\,,
\end{equation}
are obtained. 

\paragraph{Matrix with a single spherical inclusion (SPH)} We consider a soft matrix material with a stiff spherical inclusion, where the inclusion is centered and has a volume fraction of 20\%, see \cref{fig:microstructures:SPH}. We consider data from \textcite{kalina2025}, where a SPH is numerically homogenized. Therein, both material phases are described by the isotropic Neo-Hooke model 
\begin{equation}\label{eq:micro_behavior_FRE}
\begin{aligned}
\psi_{\circ}=\frac{1}{2}\Big[\mu_{\circ}(\IoneISO-\log\IthreeISO)+\frac{\lambda_{\circ}}{2}\big(\IthreeISO-\log\IthreeISO\big)\Big]\,,\\ 
\text{with}\quad
\mu_{\circ}=\frac{E_{\circ}}{2(1+\nu_{\circ})}\,,
\quad\lambda_{\circ}=\frac{E_{\circ}\nu_{\circ}}{(1+\nu_{\circ})(1-2\nu_{\circ})}\,,\quad \circ\in\{\text{matrix},\text{incl}\}\,.
\end{aligned}
\end{equation}
For the matrix material, the parameters $E_{\text{matrix}}=1\text{ MPa}$ and $\nu_{\text{matrix}}=0.4$ are employed. For the inclusion, the parameters $E_{\text{incl}}=1000\text{ MPa}$ and $\nu_{\text{incl}}=0.3$ are employed. 

For the data generation, the deformation gradient $\bF$ is sampled in a wide range of physically admissible deformation modes using a method described in \textcite[Appendix~D]{kalina2025}. In the invariant space, the sampled deformation states are within $\IoneISO\in[2.99,\,4.52]$ and $ \ItwoISO\in[2.9,\,4.9]$, thus covering a finite deformation regime. For this material, we again consider a stress--strain dataset of the form in~\cref{eq:dataset_mech_2}. We employ $800$ datapoints for {training} and roughly $3000$ datapoints for testing. 

\paragraph{Ellipticity of the homogenization data}
Polyconvex constitutive models can represent only elliptic material behavior. Thus, we want to consider the effective material behavior of the metamaterials only within their elliptic regime. It is important to note that even if the microstructure of the material is described by polyconvex (and thus elliptic) material models, the homogenized behavior might still be non-elliptic \parencite{Rudykh_deBotton_2012,KBertoldi_11_01}. To assess ellipticity of the homogenization dataset, the effective tangent $\bbA=d_{\bF}\bP$ is required. When $\bbA$ is available, the ellipticity condition~\cref{eq:ellip_mech_comp} can be numerically evaluated by checking p.s.d. of the acoustic tensor $Q_{ij}=\bbA_{i\alpha j\beta}a_{\alpha}a_{\beta}$, cf.\ \parencite[Sec.~5]{Schroeder_Neff_Balzani_2005}. While for the BCC material, no values for the tangent $\bbA$ are provided in the dataset, the excellent performance of a polyconvex PANN model applied to this dataset in \textcite{klein2022a}, along with the numerical results presented in this work, suggest that the homogenized BCC cell behaves elliptic within the considered deformation range. For the SPH material, the effective tangent $\bbA$ is available, and less than 1\% of the datapoints are non-elliptic and excluded in the further investigations.

\subsubsection{Considered PANN models}\label{sec:micro_PANNs}

For the metamaterials considered in this work, the preferred directions are known and the structural tensors and elements of the cubic symmetry group $\cO$ are chosen accordingly. The hyperparameters are chosen based on preliminary studies and in accordance to literature standards \parencite{kalina2025,klein2022b,klein2022a,Fernandez2020}. We employ the following PANN models, which were introduced in detail in~\cref{sec:PANNs}:

\paragraph{PANN models using structural tensor-based invariants} We employ both polyconvex and non-polyconvex PANN models for cubic anisotropy (cf.~\cref{sec:PANN_comp}). For the NN potentials, FFNNs with two hidden layers and 16 nodes are employed ($H=2,\,n_1=16$). In the following, such models that are polyconvex are denoted by PANN--$\cI$ and non-polyconvex models by PANN--$\cI^*$. Following~\cref{sec:inv_bases_group_symm,sec:invs_iso}, the invariants are chosen as
\begin{equation}\label{eq:cub_invar_set}
\bcI^{\text{cub}} = 
\left\{\begin{aligned}
&\big(\IoneISO,\,\ItwoISO,\,\IthreeISO,\,\JoneCUB,\,\JtwoCUB,\,\JthreeCUB,\,\JfourCUB,\,\JfiveCUB,\,\JsixCUB\big)&&\in\bbR^{9}\,, &&& (\text{PANN--}\cI^*)\\
&\big(\IoneISO,\,\ItwoISO,\,\sqrt{\IthreeISO},\,\IfourISO,\,\IoneCUB,\,\ItwoCUB,\,\IthreeCUB,\,\IfourCUB,\,\IfiveCUB,\,\IsixCUB
\big)&&\in\bbR^{10}\,, &&& (\text{PANN--}\cI)\\
\end{aligned}\right. \,.
\end{equation} 

\paragraph{PANN models based on triclinic invariants and group symmetrization} We employ the polyconvex PANN model introduced in~\cref{sec:PANN_C}, which is denoted by PANN-$\cC$. In addition, we consider a non-polyconvex PANN model that takes the non-polyconvex triclinic integrity basis introduced in~\cref{eq:K_basis_unr} as input for the NN potential. For this model, an ICNN is employed, i.e., its potential is convex in all coordinates of the right Cauchy-Green tensor. As elaborated in Remark~\ref{rem:conv_C}, this model is not polyconvex by construction. To ensure stress normalization, a projection approach is employed \parencite{Fernandez2020}. In the following, this model is denoted by PANN-$\cC^*$. For the NN potentials, FFNNs with three hidden layers and 8 nodes are employed ($H=3,\,n_{1-3}=8$). For these PANN models, the group symmetrization (cf.~\cref{sec:group_symm_intro}) is carried out with the 24 elements of the cubic symmetry group $\cO$ \parencite{spencer1971}. 

\subsubsection{Model {training}}\label{sec:application_micro_calib}

We now introduce details on the model {training}. Recall that the weights and biases of the NNs form the set of parameters $\bcP$ of the PANN model. Overall, the PANN--$\cI/\cI^*/\cC/\cC^*$ models have $464\,/\,448\,/\,272\,/\,\allowbreak 208$ trainable parameters.

\paragraph{Formulation of the loss function} The overall dataset is split into a {training} dataset consisting of $m_{\text{tr}}$ datapoints and a test dataset consisting of $m_{\text{te}}=m-m_{\text{tr}}$ datapoints. To fit the PANN model parameters $\bcP$ to a given dataset $\cD$, a loss function $\mathfrak{L}$ is minimized. We consider the loss function given as the mean squared error (MSE) of the stress predictions 
\begin{equation}\label{eq:stress_loss_micro_P}
\begin{aligned}
        \mathfrak{L}(\bcP)=\frac{1}{m_{\text{tr}}}\sum_{i=1}^{m_{\text{tr}}}\frac{1}{9\star}\left\|{^i\bP}-{^i\bP}^{\text{model}}({^i\bF},\,\bcP)\right\|^2\,,
\end{aligned}
\end{equation}
where $\star$ is introduced to make the loss function dimensionless. For the BCC material, we set $\star=\text{hPa}$, while for the SPH material, we set $\star=\text{MPa}$. All PANN models are solely {trained} on first derivatives of the strain energy, even if data for the second derivative is available. The {training} of the PANN models through their gradients is referred to as Sobolev training \parencite{vlassis2021,Czarnecki2017SobolevTF}. Since we focus on the prediction of the derivatives of the energy potentials in this work, we do not consider energy values in the {training} process. This decision is supported by the observation that including energy data in the {training} process does not enhance the prediction quality of the PANN model's derivatives \parencite{klein2022a}. 

\begin{table}[t]
    \centering
    \caption{$\log_{10}$MSE for PANN models applied to the BCC and the SPH data. Each model is {trained} five times, and the results for the model instance with the lowest {training} MSE are provided.}
        {\small
    \begin{tabular}{llp{0.075cm}lp{0.075cm}lp{0.075cm}lp{0.075cm}l}
    \toprule
  &&  & \textbf{PANN}--$\cC$ && \textbf{PANN}--$\cI$ && \textbf{PANN}--$\cC^*$ && \textbf{PANN}--$\cI^*$ \\
    \midrule
    \midrule
    \textbf{BCC} & {training} && -1.8 && -0.18 && -2.1 && -2.1 \\
     & test && -0.48 && 0.23 && -0.91&& -0.07 \\
     \midrule
    \textbf{SPH} & {training} && -2.45 && -1.75 && -3.53 && -3.24 \\
     & test && -2.38 && -1.7 && -3.38 && -3.12 \\    \bottomrule
    \end{tabular}
    }
    \label{tab: loss}
\end{table}

\begin{figure}[t!]
\centering
\includegraphics[width=\textwidth]{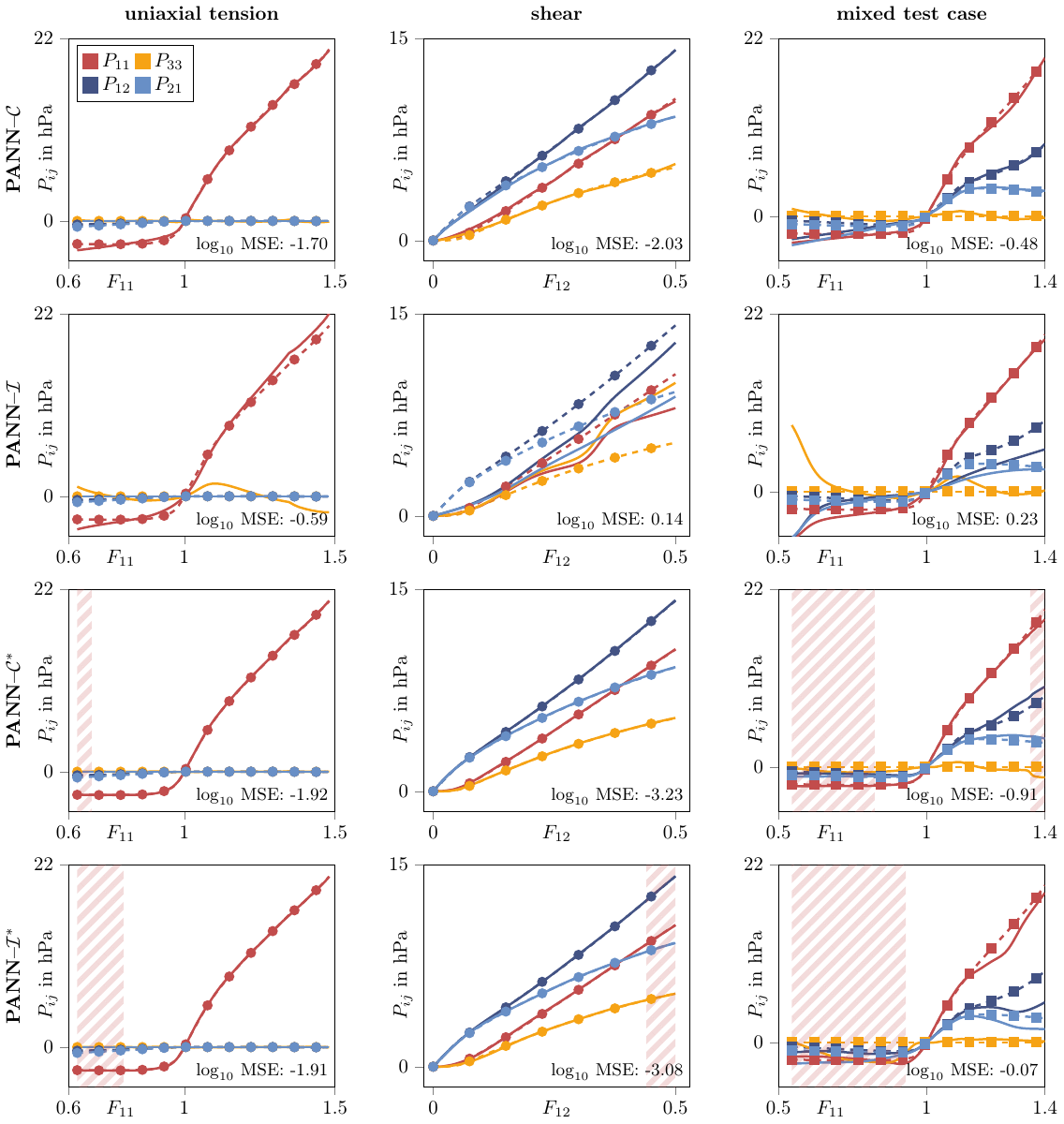}
\caption{Performance of different PANN models for the BCC data. Circles denote the {training} data, squares the test data, and solid lines the model predictions. Shaded areas indicate a loss of ellipticity.
} 
\label{fig:BCC_pc}
\end{figure}

\paragraph{Implementation and {training} details} The PANN models are implemented and {trained} using the KLAX library (\url{https://drenderer.github.io/klax/}). The parameter optimization is performed with the ADAM optimizer, employing the full {training} dataset, a mini-batch size of 32 with batch {shuffling}, without loss weighting, for $250\,000$ steps, and with a learning rate of $5 \cdot 10^{-3}$. For the SPH material, this is followed by a re-{training} with the SLSQP optimizer, employing the full {training} dataset, a mini-batch size of 32 with batch {shuffling}, without loss weighting, for a maximum of 750 iterations, and with a solver tolerance of $5\cdot 10^{-32}$. In each investigation, the PANN models are {trained} five times, and the model with the lowest {training} loss is used for evaluation. By that, we account for the random initialization of the NN parameters and stochastic effects in the optimization process. {The test dataset was not used for model selection or hyperparameter tuning and was evaluated only after the final model has been selected, to ensure that it provides an independent assessment of the model's predictive capability for unseen loading cases.} For the volumetric growth term (cf.~\cref{eq:PANN_growth}), the parameter $\alpha$ is set to $1 \cdot 10^{-5}$ hPa (MPa) for the BCC (SPH) material. In this way, it is ensured that the influence of the growth term within the deformations of the available {training} and test data is negligible, and that it gets active only once large volumetric compression occurs. The stress normalization terms are included in the model {training} process. 

\subsection{PANN model evaluation}\label{sec:evaluation}

In the following, we investigate the performance of different PANN constitutive models for the homogenization data. Note that the BCC data is very sparse and only covers a few deformation modes, while the SPH data contains significantly more datapoints and includes more general deformation modes.

\paragraph{Application to the BCC data} We apply all PANN models to the BCC data. In~\cref{tab: loss}, the loss values of the best model instances {with the lowest training loss} are provided, while in~\cref{fig:BCC_pc}, the stress predictions for the different {trained} models are visualized. Due to buckling phenomena in the BCC structure's beams, its stress response is highly nonlinear. In particular, the stress curve shows a distinct change of slope around the undeformed state for both uniaxial tension and the mixed test case, indicating a highly different behavior for tension and compression.
The polyconvex PANN--$\cC$ model has an excellent performance on both {training} cases and a decent prediction of the mixed case. In contrast, the polyconvex PANN--$\cI$ model shows a moderate performance for the uniaxial load case and fails to predict both the shear and the mixed case. While for the BCC material, the ellipticity of the effective material behavior cannot be evaluated directly, the excellent performance of the PANN--$\cC$ suggests that the BCC material remains elliptic within the considered deformation modes. 
Both non-polyconvex PANN models show an excellent stress prediction for the {training} cases. For the test case, the PANN--$\cC^*$ model also has an excellent performance, while the PANN--$\cI^*$ model has only a moderate performance. Notably, these models did not adopt the ellipticity of the BCC material, but show a pronounced loss of ellipticity across the investigated deformation scenarios. This is because the employed {training} data is very sparse. This loss of ellipticity would cause convergence issues in numerical simulations and therefore represents a significant limitation of the models.

\begin{figure}[t!]
\centering
\includegraphics[width=\textwidth]{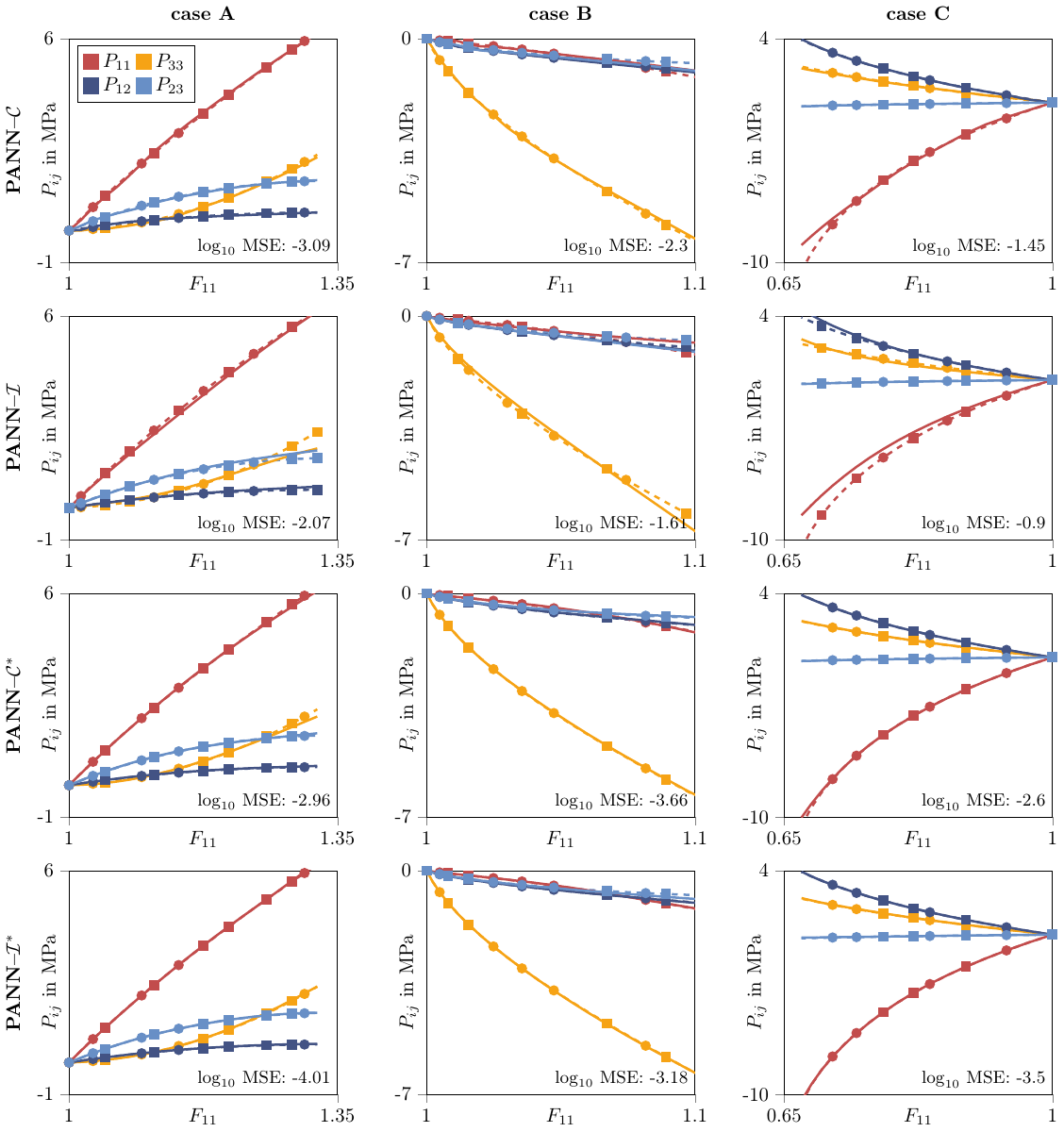}
\caption{Performance of different PANN models for the SPH data. Circles denote the {training} data, squares the test data, and solid lines the model predictions.} 
\label{fig:SPH}
\end{figure}

\paragraph{Application to the SPH data} We also apply all PANN models to the SPH data. In~\cref{tab: loss}, the loss values of model instances {with the lowest training loss} are provided, while in~\cref{fig:SPH}, three representative load cases are visualized. Overall, the PANN--$\cI$ has exhibits the lowest performance, while the PANN--$\cC$ achieves a better performance, but still shows visible deviations from the ground truth material behavior. Both non-polyconvex models show an excellent performance in predicting the stress values. Recall that we considered only elliptic datapoints for the SPH material. Since we employed sufficient {training} data, the non-polyconvex PANN models could adopt the ellipticity of the SPH material, and are elliptic for all considered datapoints.

\paragraph{Discussion} For cubic anisotropy, both considered polyconvex PANN models fulfill all common constitutive conditions a priori. However, the PANN--$\cC$ outperforms the PANN--$\cI$ for the considered materials. This demonstrates that the specific choice of polyconvex model formulation can largely influence the model performance. For the BCC material with limited data, the PANN--$\cC$ further outperforms its non-polyconvex counterparts by providing a stable, i.e., elliptic, representation, whereas the latter exhibit a pronounced loss of ellipticity. In contrast, for the SPH material, the polyconvex models achieve only moderate performance, while the non-polyconvex models accurately captures the material response and, at least for the available data, learns an elliptic representation. This demonstrates that, given suitable {training} data, also non-polyconvex models can be employed when more flexibility is required. When only sparse {training} data is available, convexity-promoting loss terms could be included in the {training} process to improve the performance of non-polyconvex PANN models \parencite{kalina2024a}. However, non-polyconvex models come at the cost of losing ellipticity guarantees and can reduce the model generalization \parencite{Klein2026a,kalina2024a}. Consequently, non-polyconvex models should be employed with caution and only when strictly necessary.


\section{Conclusion}\label{sec:conc}

In this work, we addressed some open challenges for polyconvex anisotropic hyperelasticity and physics-augmented neural network (PANN) constitutive modeling. {We proposed a new constitutive modeling framework for polyconvex anisotropic hyperelasticity based on triclinic invariants and group symmetrization, and applied it to a PANN model. For finite symmetry groups such as the cubic one, this PANN model fulfills all common mechanical constitutive conditions a priori.} Moreover, we contributed to the theory of polyconvex invariant theory, which is the basis of many polyconvex constitutive models. First, we provided an introduction into the construction of polyconvex integrity or functional bases. We proposed a group symmetrization-based method for the construction of polyconvex invariants for finite symmetry groups. Based on this, we derived a polyconvex integrity basis for the tetragonal $\cD_4$ group and a polyconvex functional basis for the cubic $\cO$ symmetry group. We furthermore {revised} polyconvex integrity bases from the literature for the triclinic $\cC_1$, isotropic $\cK$, transversely isotropic $\cD_{\infty}$, monoclinic $\cC_2$, and rhombic $\cD_2$ symmetry groups. 

We compared the performance of a polyconvex PANN model based on invariants constructed via structural tensors and a polyconvex PANN model based on triclinic invariants and group symmetrization, by applying them to homogenization data of cubic metamaterials. Although, for this symmetry group, both approaches fulfill all common constitutive conditions of hyperelasticity a priori, the latter model outperforms the former. This demonstrates that the specific choice of polyconvex model formulation can have a pronounced influence on the model performance. 

In future works, the polyconvex PANN model based on triclinic invariants and group symmetrization should be extended to multiphysical material behavior, for which currently only approaches using structural tensor-based invariants are available \parencite{klein2022b,Fuhg_Jadoon_Weeger_Seidl_Jones_2024,kalina2024a,ORTIGOSA2025117741}. While for isotropy, a polyconvex constitutive model with universal approximation properties was recently proposed \parencite{wiedemann2023characterizationpolyconvexisotropicfunctions,Geuken_Kurzeja_Wiedemann_Mosler_2025}, the formulation of such models for anisotropy remains an open challenge. Closely related to this, polyconvex integrity or functional bases need to be established for all symmetry groups. The group symmetrization-based approach proposed in this work may provide a promising pathway to achieve this. Despite the advances presented in this work, many challenges remain, and polyconvex anisotropic hyperelasticity will continue to be a demanding and fascinating field of research.


\vspace*{3ex}

\noindent
\textbf{CRediT authorship contribution statement.} 
\textbf{D.K.\ Klein:} Conceptualization, Methodology, Software, Formal analysis, Investigation, Writing -- original draft, Writing -- review and editing, Visualization. 
\textbf{K.A.\ Kalina:} Conceptualization, Methodology, Formal analysis, Writing -- review and editing. 
\textbf{R. Ortigosa:} Conceptualization, Funding acquisition, Writing -- review and editing. 
\textbf{J. Mart\'inez-Frutos:} Conceptualization, Funding acquisition, Writing -- review and editing.   
\textbf{M. Kästner:} Funding acquisition, Writing -- review and editing. 
\textbf{O. Weeger:} Conceptualization, Funding acquisition, Writing -- review and editing. 
\vspace*{1ex}

\noindent
\textbf{Conflict of interest.} The authors declare that they have no conflict of interest.
\vspace*{1ex}

\noindent
\textbf{Acknowledgment.} 
The authors would like to thank Brain M. Riemer for helpful discussions and explanations on the theory of invariants. D.K.\ Klein and O.\ Weeger acknowledge the financial support provided by the Deutsche Forschungsgemeinschaft (DFG, German Research Foundation, project number 492770117) and the Graduate School of Computational Engineering at TU Darmstadt. K.A.\ Kalina and M.\ Kästner want to thank the DFG for the financial support within the Research Training Group
GRK 2868 D${}^3$–Project Number 493401063. R.\ Ortigosa and J.\ Mart\'inez-Frutos acknowledge the support of grant PID2022-141957OA-C22 funded by MICIU/AEI/10.13039/501100011033 and by ``ERDF A way of making Europe''. They also acknowledge the support provided by the Autonomous Community of the Region of Murcia, Spain through the programme for the development of scientific and technical research by competitive groups (21996/PI/22), included in the Regional Program for the Promotion of Scientific and Technical Research of Fundacion Seneca - Agencia de Ciencia y Tecnologia de la Region de Murcia.
\vspace*{1ex}

\noindent
\textbf{Data availability.}
The authors have no data to share.
\vspace*{1ex}

\noindent
\textbf{AI declaration.} This manuscript was written independently and linguistically revised with the help of ChatGPT, Grammarly, and DeepL.


\appendix
\numberwithin{equation}{section} 



\section{Notation}\label{app:notation}

Throughout this work, tensor spaces of rank greater than zero are denoted by
\begin{equation}
\cL_n:=\underbrace{\bbR^3\otimes\cdots\otimes\bbR^3}_{n-\text{times}}\,\,\forall n\in\bbN^+\,,
\end{equation}
where $\bbR^3$, $\bbN^+$, and $\otimes$ denote the Euclidean vector space in $\bbR^3$, the set of natural numbers excluding zero, and the dyadic product, respectively. First, second, and fourth order tensors are denoted by $\ba\in\cL_1$, $\bA\in\cL_2$, and $\bbA\in\cL_4$, respectively. 

Tensor compositions and contractions are denoted by $(\bA\cdot\bB)_{ij}=A_{ik}B_{kj}$, $(\bA\cdot\bb)_{i}=A_{ik}b_{k}$, $\ba\cdot\bb=a_ib_i$, $\bA:\bB=A_{ij}B_{ij}$, and $(\bbA:\bA)_{ij}=\bbA_{ijkl}A_{kl}$, respectively, where Einstein's summation convention is applied. 
The tensor cross product operator $\Cross$ is defined as $(\bA\Cross\bB)_{iI}=\mathcal{E}_{ijk}\mathcal{E}_{IJK}A_{jJ}B_{kK}$, where $\mathcal{E}_{IJK}$ is the permutation symbol. 
The transpose and inverse of a second order tensor $\bA$ are denoted by $\bA^T$ and $\bA^{-1}$, while trace, determinant, and cofactor are denoted by $\tr\bA$, $\det\bA$, and $\cof\bA:=\det(\bA)\bA^{-T}$, respectively. Norms of tensors of order one and two are given by $\norm{\ba}=\sqrt{a_ia_i}$ and $\norm{\bA}=\sqrt{A_{ij}A_{ij}}$, respectively.
The  first Fr\'echet derivative of a function $f$ w.r.t.\ $\bA$ is denoted by $\partial_{\bA}f$, the second Fr\'echet derivative w.r.t.\ $\bA$ and $\bB$ is denoted by $\partial^2_{\bA\bB}f$. Partial derivatives are denoted by $\partial_{\bA}f$, total derivatives are denoted by $d_{\bA}f$.
The gradient operator is denoted by $\nabla=\partial_{x_k}(\bullet) \be_k$ and the divergence operator by $\nabla\cdot$. A zero as a subscript next to the nabla symbol, i.e., $\nabla_0$, indicates that the operation is carried out in the reference configuration.

As common in the literature, we do not explicitly distinguish between the algebraic structure of a group and the associated set. For example, $\cG$ can stand for the group given by the tuple $(\cG,\,\cdot)$ or the set $\cG$ of tensors involved therein. We consider in the $\bbR^3$ the general linear group $\GL^+(3):=\big\{\bA \in\allowbreak \;\cL_2\,\rvert\,\allowbreak \det \bA >0\big\}$, the space of symmetric positive definite tensors $\SYM^+(3):=\big\{\bA \in\allowbreak \;\cL_2\,\rvert\,\allowbreak \bA=\bA^T,\,\ba\cdot\bA\cdot\ba>0\,\forall\ba\in\cL_1\backslash{\{\bnull}\}\big\}$, the orthogonal group $\operatorname{O}(3):=\big\{\bA \in\allowbreak \cL_2\;\rvert\allowbreak \;\bA^T\cdot\bA=\bI\big\}$, and the special orthogonal group $\SO(3):=\big\{\bA \in\allowbreak \cL_2\;\rvert\allowbreak \;\bA^T\cdot\bA=\bI,\;\det \bA =1\big\}$, where $\bI\in\cL_2$ denotes the second order identity tensor. 
To enhance readability, function arguments are omitted throughout this work unless their inclusion is required for clarity.



\section{{Additional information on invariant-based modeling}}\label{app:add_inf_invs}

{In this appendix, we provide additional results on invariant-based modeling. While the following results are, in large parts, well-established, we collect them here in a unified notation for the convenience of the reader. We provide a coordinate-dependent polyconvex integrity basis for triclinic anisotropy in~\cref{app:invs_tri_coord}. For the derivation of polyconvex integrity bases, we follow the procedure introduced in~\cref{sec:general_construction}. In~\cref{app:invs_established_struct}, we give by a brief introduction to the construction of invariants via structural tensors and isotropic tensor functions, and provide corresponding polyconvex integrity bases for several symmetry groups. In~\cref{app:complete_pc_bases}, we discuss the challenge of formulating polyconvex integrity or functional bases for anisotropic material behavior. Finally, in~\cref{app:inv_relations}, we provide the relations between different invariant bases used throughout this manuscript.}


\subsection{Coordinate-dependent polyconvex integrity basis for triclinic anisotropy}\label{app:invs_tri_coord}

{We now provide a polyconvex integrity basis for triclinic anisotropy in coordinate-dependent form, which we employ in the constitutive modeling framework in~\cref{sec:aniso_hyperel_tric_group}.} A polyconvex integrity basis for the triclinic symmetry group $\cC_1$ was first proposed by \textcite[Sec.~3.2]{Schroeder2008}.

\paragraph{(A) Minimal integrity basis} The coordinate-dependent invariants
\begin{equation}\label{eq:K_basis_unr}
\begin{aligned}
 \Kone=C_{11}\,,\qquad \Ktwo=C_{22}\,,\qquad\Kthree=C_{33}\,,\qquad\Kfour= C_{12}\,,\qquad\Kfive=C_{13}\,,\qquad\Ksix=C_{23}\,,
    \end{aligned}
\end{equation}
proposed by~\parencite{smith1962} form a minimal integrity basis for the triclinic $\cC_1$ group. The triclinic group is the most general material symmetry group, as reflected by the fact that~\cref{eq:K_basis_unr} contains all six independent coordinates of the right Cauchy-Green tensor. Thus, with triclinic anisotropy, arbitrary (an-)isotropic material behavior can be represented. 

\paragraph{(B) Polyconvex integrity basis}

Not all coordinates of $\bC$ are polyconvex. In particular, while the main diagonal components of $\bC$ are polyconvex, the remaining components are not elliptic and thus not polyconvex. The difference between main and off-diagonal elements of $\bC$ becomes evident for the examples
\begin{equation}\label{eq:C12}
\begin{aligned}
           \bC:(\be^{(1)}\otimes\be^{(1)})&=  C_{11}=F_{11}^2+F_{21}^2+F_{31}^2\,,
    \\
     \bC:(\be^{(1)}\otimes\be^{(2)})&=     C_{12}=F_{11}F_{12}+F_{21}F_{22}+F_{31}F_{32}\,,
\end{aligned}
\end{equation}
where $\be^{(i)}\in\bbR^3$ denote the basis vectors of the employed Cartesian coordinate system. The component $C_{11}$ is convex in the deformation gradient $\bF$. In contrast to that, $C_{12}$ is not elliptic and thus not polyconvex, which can be shown by evaluating the ellipticity condition (cf.~\cref{eq:ellip_mech_comp}) for, e.g., $[\ba]=(0,0,1)^T,\,[\bb]=(1,-1,0)^T$. The difference is that $C_{11}$ only contains square functions of individual $F_{ij}$, which is a convex operation, whereas $C_{12}$ contains multiplicative terms of different $F_{ij}$, which generally does not preserve convexity \parencite{Schroeder2003}. Thus, the triclinic integrity basis in~\cref{eq:K_basis_unr} is not polyconvex. To construct a polyconvex triclinic integrity basis, we consider the polyconvex projection
\begin{equation}\label{eq:C_comp}
\bC:\bA^{(i)}\in\bbR\quad\text{with}\quad\bA^{(i)}={\bd}^{(i)}\otimes{\bd}^{(i)}\quad \text{for}\quad{\bd}^{(i)}\in\bbR^3\,,\,\,\normm{{\bd}^{(i)}}=1\,.
\end{equation}
By inserting the vectors
\begin{equation}\label{eq:cu}
    \begin{aligned}
\be^{(1)},\,\be^{(2)},\,\be^{(3)},\,(\be^{(1)}+\be^{(2)})/\sqrt{2},\,(\be^{(1)}+\be^{(3)})/\sqrt{2},\,(\be^{(2)}+\be^{(3)})/\sqrt{2}\,,
    \end{aligned}
\end{equation}
into~\cref{eq:C_comp}, we construct the coordinate-dependent invariants
\begin{equation}\label{eq:K_basis_C}
\begin{aligned}
& \Cone=C_{11}\,, \qquad\Ctwo=C_{22}\,,\qquad\Cthree=C_{33}\,,\qquad
&&\Cfour= (C_{11}+C_{22}+2C_{12}) /{2}\,,
\\
&\Cfive=(C_{11}+C_{33}+2C_{13})/{2}\,,&&\Csix=(C_{22}+C_{33}+2C_{23})/{2}\,,
    \end{aligned}
\end{equation}
which form a polyconvex integrity basis for the triclinic $\cC_1$ group, as it is possible to express $k_{1-6}$ as polynomials in $c_{1-6}$. The relation between both is linear, and we leave the calculations to the interested reader.

\paragraph{(C) Extension of the polyconvex integrity basis}

In a polyconvex modeling framework, it seems very natural to also consider the coordinates of the cofactor of the right Cauchy-Green tensor $\bG$ and the determinant of the deformation gradient $J$. This is motivated by the fact that $J$
 and the coordinates of $\bG$ are not convex functions of $\bC$, and therefore cannot be represented within a convex formulation based solely on $\bC$. Thus, including $\bG$ and $J$ can improve the flexibility of the constitutive model. Similar to~\cref{eq:K_basis_C}, we construct the coordinate-dependent invariants of $\bG$
\begin{equation}\label{eq:K_basis_G}
\begin{aligned}
& \Cseven=G_{11}\,, \qquad\Ceight=G_{22}\,,\qquad\Cnine=G_{33}\,,\qquad
&& \Cten= (G_{11}+G_{22}+2G_{12}) /{2}\,,
 \\
 &\Celeven=(G_{11}+G_{33}+2G_{13})/{2}\,,&&\Ctwelve=(G_{22}+G_{33}+2G_{23})/{2}\,,
    \end{aligned}
\end{equation}
where polyconvexity can be shown by replacing $\bC$ with $\bG$ in~\cref{eq:C_comp}. Moreover, we consider
\begin{equation}\label{eq:K_basis_J}
\Cthirteen=J\,,\qquad\Cfourteen=-J\,.
\end{equation}
Including both $J$ and $-J$ is motivated as follows. Invariants such as $C_{11}$ are nonlinear functions of the arguments of the polyconvexity condition, in particular, $\det\bC=J^2$ is a nonlinear function in $J$. Thus, the Hessian operator of the invariant-based potential consists of both a constitutive type term and a geometric type term (cf.~\cref{eq:hess_inv}), and the potential has to be both convex and monotonically increasing in such invariants to preserve their polyconvexity. In contrast to that, $J$ is the only invariant quantity directly included in the polyconvexity condition. Being linear in $\bxi$, the constitutive type term vanishes for $J$ (cf.~\cref{eq:hess_inv}). Thus, the potential $\psi$ must not necessarily be monotonically increasing in this invariant. This is pragmatically taken into account by adding $-J$ in the set of arguments for polyconvex models, which allows to formulate potentials which are convex and monotonically increasing in all of their arguments without being overly restrictive in $J$.\footnote{It would also be possible to use an adapted NN architecture that is convex in $J$ and convex and monotonic in the remaining invariants \parencite{KALINA2026118892,jadoon2025inputspecificneuralnetworks}.} Overall, this simplifies the formulation of polyconvex NN constitutive models, as it enables to express the potential as a convex and monotonic function in \emph{all} of the invariants (cf.~\cref{sec:PANN_comp}).
We consider coordinate-dependent triclinic invariants only for polyconvex constitutive models, for which we employ the invariants
\begin{equation}\label{eq:C_pc}
\bcI^{\text{tri}}=\big(\Cone,\,\Ctwo,\,\Cthree,\,\Cfour,\,\Cfive,\,\Csix,\,\Cseven,\,\Ceight,\,\Cnine,\,\Cten,\,\Celeven,\,\Ctwelve,\,\Cthirteen,\,\Cfourteen\big)\in\bbR^{14}\,.
\end{equation}
The elements of the tuple $\bcI^{\text{tri}}$ form an integrity basis, but it is not minimal. However, since the amount of invariants remains moderate even if some additional invariants are added to the minimal set, non-minimality poses no problem for the constitutive modeling framework applied in this work.


\subsection{Polyconvex integrity bases using structural tensors}\label{app:invs_established_struct}

{In the following, we briefly revise the theoretical background of invariants constructed via structural tensors and isotropic tensor functions. Subsequently, we provide established polyconvex integrity bases, covering the isotropic $\cK$ (iso), transversely isotropic $\cD_{\infty}$ (ti), triclinic $\cC_1$ (tri),  monoclinic $\cC_{2}$ (mon), and rhombic $\cD_{2}$ (rho) groups, which we abbreviate as indicated in the brackets. Note that all of the invariant bases introduced in the following are \emph{integrity} bases, cf.~\cref{sec:general_construction}.}

\medskip

The formulation of invariants for anisotropic materials is commonly based on structural tensors \parencite{Zheng1993}. For the symmetry groups considered in this work, second-order structural tensors $\bM^{\square}\in\cL_2$ and fourth-order structural tensors $\bbM^{\square}\in\cL_4$ are sufficient to characterize the material behavior. The tuple of structural tensors $\cS^{\square}=\big(\bM^{\square}_{(1)},\,\dotsc\bM^{\square}_{(m)},\,\bbM^{\square}_{(1)},\,\dotsc\bbM^{\square}_{(n)}\big)$ belongs to the symmetry group $\cG^{\square}\subseteq\operatorname{O}(3)$ if it satisfies the conditions
\begin{equation}\label{eq:def_struct}
\begin{aligned}
    &\bQ\star\cS^{\square}=\cS^{\square}\,\,\,
    \forall\, \bQ\in \cG^{\square}
\,\,\,\land\,\,\,
\bQ\star\cS^{\square}\neq\cS^{\square}
 \,\,\,   \forall\, \bQ\not\in \cG^{\square}\quad \Leftrightarrow\quad\bQ\star\cS^{\square}=\cS^{\square}\implies\bQ\in \cG^{\square}
 \\
 \text{with}\quad &\bQ\star\cS^{\square}=\big(\bQ\star\bM^{\square}_{(1)},\,\dotsc,\,\bQ\star\bM^{\square}_{(m)},\,\bQ\star\bbM^{\square}_{(1)},\,\dotsc,\,\bQ\star\bbM^{\square}_{(n)}\big)
 \\
 \text{and}\quad\big(&\bQ\star\bM\big)_{ij}=Q_{ia}Q_{jb} M_{ab},\quad
 \left(\bQ\star\bbM\right)_{ijkl}=Q_{ia}Q_{jb}Q_{kc}Q_{ld} M_{abcd}\,.
\end{aligned}
\end{equation}
The isotropicization theorem states that the material symmetry condition (cf.~\cref{sec:const_cond}) for the symmetry group $\cG$ can be alternatively expressed as 
\begin{equation}\label{eq:isotropicization}
\begin{aligned}
W\big(\bF,\,\cS^{\square}\big)=W\big(\bF\cdot\bQ^T,\,\bQ\star\cS^{\square}\big) 
\qquad
\forall(\bF,\,\bQ)\in\GL^+(3)\times\operatorname{O}(3)\,.
\end{aligned}
\end{equation}
This means that by including structural tensors in the arguments of the hyperelastic potential, it can be expressed as an isotropic tensor function even if the considered material behavior is anisotropic \parencite{itskov2015,Haupt2002}. Representation theorems for scalar-valued tensor functions provide expressions for invariants of the right Cauchy-Green tensor $\bC$ and the tuple of structural tensors $\cS^{\square}$ \parencite{boehler1977,itskov2015,Xiao1996OnIE}. Based on this, different tuples of invariants can be derived. {Integrity bases for all common symmetry groups of hyperelasticity are provided in \textcite{riemer2025}.}

In the following, we provide established polyconvex integrity bases. Thereby, we denote the preferred direction(s) of a material in the reference configuration by orthonormal vectors
\begin{equation}\label{eq:pref_dir}
\bn_{0}^{(i)}\in\bbR^3\qquad
\text{with}\qquad\bn_{0}^{(i)}\cdot\bn_{0}^{(j)}=\delta_{ij}\,.
\end{equation} 
If the material has only a single preferred direction, the superscript $(i)$ is omitted. In~Appendix~\ref{app:inv_relations}, we provide the functional relations between the different invariant bases. By this, we verify that the following polyconvex invariants belong to the considered symmetry group, and that all of them form integrity bases.

\subsubsection{Isotropic material behavior}\label{sec:invs_iso}

\paragraph{(A) Minimal integrity basis}

For the isotropic symmetry group $\cK$, no structural tensors are required, i.e., $\cS^{\text{iso}}=\emptyset$.\footnote{Note that the identity tensor $\bI$ could be employed as structural tensor for isotropic material behavior.} The three invariants
\begin{equation}\label{eq:invs_iso_standard}
    \JoneISO = \tr\bC\,,\qquad \JtwoISO=\tr\bC^2\,,\qquad \JthreeISO=\tr\bC^3\,,
\end{equation}
form a minimal integrity basis for the isotropic $\cK$ group.

\paragraph{(B) Polyconvex integrity basis}

While the invariants in~\cref{eq:invs_iso_standard} are polyconvex~\parencite[Sec.~3.2.1]{Schroeder2003}, the isotropic polyconvex invariants
\begin{equation}\label{eq:invs_iso}
    \IoneISO = \tr\bC=\norm{\bF}^2\,,\qquad \ItwoISO=\tr\bG=\norm{\bH}^2\,,\qquad \IthreeISO=\det\bC=J^2\,,
\end{equation}
are a more popular choice in constitutive modeling. The invariants in~\cref{eq:invs_iso} form a polyconvex integrity basis for the isotropic $\cK$ group as $I^{\text{iso}}_{1-3}$ can represent $J^{\text{iso}}_{1-3}$ as polynomials, cf.~\textcite[Eq.~(4.29)]{itskov2015}, \textcite[Eq.~(3.18)]{Schroeder2003}, and~App.~\ref{app:inv_relations_iso}. A polyconvex integrity basis for the isotropic $\cK$ group was first proposed by \textcite[Sec.~8]{Ball1976}.

\paragraph{(C) Extension of the polyconvex integrity basis} 

Instead of $\IthreeISO$, it can be convenient to employ $J=\sqrt{\IthreeISO}$ for NN-based constitutive modeling and furthermore include $\IfourISO=-J$, which we already motivated below~\cref{eq:K_basis_J}. 

\subsubsection{Transversely isotropic material behavior}\label{sec:invs_TI}

Two transversely isotropic symmetry groups are relevant in anisotropic hyperelasticity \parencite{riemer2025}. In this work, we only consider the transversely isotropic $\mathcal{D}_{\infty}$ group.

\paragraph{(A) Minimal integrity basis}

For the transversely isotropic symmetry group $\mathcal{D}_{\infty}$, a single second-order structural tensor is sufficient to characterize the materials anisotropy, e.g.,
\begin{equation}\label{eq:struct_ti}
\cS^{\text{ti}}=\big(\bMti\big)\qquad\text{with}\qquad\bMti =\bn_0\otimes\bn_0\,.
\end{equation} 
Together with the isotropic invariants $J_{1-3}^{\text{iso}}$ introduced in~\cref{sec:invs_iso}, the invariants
\begin{equation}\label{eq:invs_ti_unr}
\JoneTI=\tr\bC\cdot\bMti\,,\qquad \JtwoTI=\tr\bC^2\cdot\bMti\,,
\end{equation}
form a minimal integrity basis for the transversely isotropic $\cD_{\infty}$ group.

\paragraph{(B) Polyconvex integrity basis}

While $\JoneTI$ is already a polyconvex invariant, $\JtwoTI$ is not elliptic and thus not polyconvex \parencite{merodio2006}. This demonstrates that invariants obtained from representation theorems, which are based on powers of $\bC$, quickly lead to non-polyconvex formulations. To obtain a polyconvex invariant that includes the information of $\JtwoTI$, the Cayley-Hamilton theorem \parencite{Holzapfel2000,itskov2015} can be employed:
\begin{equation}\label{eq:G_cayley_hamilton}
\bG=\bC^2-\IoneISO\bC+\ItwoISO\bI\,.
\end{equation}
Notably, the cofactor $\bG=\cof\bC$ is a quadratic expression in $\bC$. Furthermore, the cofactor of the deformation gradient $\bH$ is included in the arguments of the polyconvexity condition, and $\bH$ is closely related to $\bG$ via $\bG=\bH^T\cdot\bH$. Based on this observation, a variety of polyconvex invariants can be formulated for anisotropic material behavior \parencite{Schroeder2003}. For transverse isotropy, the two polyconvex invariants
\begin{equation}\label{eq:invs_ti_pc}
\IoneTI=\tr\bC\cdot\bMti=\norm{\bF\cdot\bMti}^2\,,\qquad \ItwoTI=\tr\bG\cdot\bMti=\norm{\bG\cdot\bMti}^2\,,
\end{equation}
can be constructed \parencite{Schroeder2003}.\footnote{Formulating invariants in terms of the arguments of polyconvexity also provides a geometric interpretation of such invariants \parencite[Sec.~3.2.2]{Schroeder2003}.} Together with the isotropic invariants $I_{1-3}^{\text{iso}}$ introduced in~\cref{sec:invs_iso}, the invariants in \cref{eq:invs_ti_pc} form a polyconvex integrity basis for the transversely isotropic $\cD_{\infty}$ group. This is because $I^{\text{ti}}_{1,2}$ and $I^{\text{iso}}_{1-3}$ can represent $J^{\text{iso}}_{1-3}$ and $J^{\text{ti}}_{1,2}$ as polynomials, cf.~\textcite[Eq.~(3.46)]{Schroeder2003}, App.~\ref{app:inv_relations_iso}, and~App.~\ref{app:inv_relations_ti}. A polyconvex integrity basis for the transversely isotropic $\cD_{\infty}$ group was first proposed by \textcite[Sec.~3.2.2]{Schroeder2003}.

\paragraph{(C) Extension of the polyconvex integrity basis}

As proposed by \textcite{Schroeder2003}, the additional polyconvex invariants
\begin{equation}\label{eq:invs_ti_2}
    \IthreeTI=\IoneISO-\IoneTI\,,\qquad \IfourTI=\ItwoISO-\ItwoTI\,,
\end{equation}
can be constructed. The use of the additional invariants $\IthreeTI$ and $\IfourTI$ can diminish the restrictions that polyconvexity poses on the energy potential. In particular, a sufficient condition for polyconvexity is that the potential is convex and monotonic in polyconvex invariants, cf.~\cref{sec:inv_conv}. However, as the polyconvex invariants in \cref{eq:invs_ti_2} demonstrate, the monotonicity condition is not necessary. E.g., considering the potential $\psi(\IoneTI,\,\IthreeTI)$, its total derivative w.r.t.\ $\IoneTI$ yields
\begin{equation}
\begin{aligned}
    d_{\IoneTI}\psi(\IoneTI,\,\IthreeTI)&=\partial_{\IoneTI}\psi(\IoneTI,\,\IthreeTI)+\partial_{\IthreeTI}\psi(\IoneTI,\,\IthreeTI) d_{\IoneTI}\IthreeTI
=\partial_{\IoneTI}\psi(\IoneTI,\,\IthreeTI)-\partial_{\IthreeTI}\psi(\IoneTI,\,\IthreeTI) \,,
    \end{aligned}
\end{equation}
which demonstrates that including the additional invariant $\IthreeTI$ allows for decreasing potentials in $\IoneTI$ whenever the condition
\begin{equation}\label{eq:exa_ti_1}
    \partial_{\IoneTI}\psi(\IoneTI,\,\IthreeTI)<\partial_{\IthreeTI}\psi(\IoneTI,\,\IthreeTI) \,,
\end{equation}
holds. In particular, \cref{eq:exa_ti_1} can be fulfilled even if the potential is convex and monotonic in all invariants. Thus, using these additional invariants, at least in parts, allow for decreasing functional relationships of the potential $\psi$ in $\IoneTI$ and $\ItwoTI$, thus providing the constitutive model with more flexibility. Note that these additional invariants are not required when non-polyconvex models are considered, since, in this case, the potential is not subject to convexity or monotonicity conditions. 

\begin{remark}\label{rem:add_invs}
    The invariants in~\cref{eq:invs_ti_2} can be expressed as
    \begin{equation}
        \tr\bA-\tr\bA\cdot\left(\bn_0\otimes\bn_0\right)=\tr\bA\cdot\left(\bI-\bn_0\otimes\bn_0\right)\,,
    \end{equation}
    where $\bA$ represents either $\bC$ or $\bG$. Then, polyconvexity of those invariants immediately follows from the observation that $\bI-\bn_0\otimes\bn_0$ is p.s.d. \parencite[Lemma~A.12]{Schroeder2003}.
\end{remark}

\subsubsection{Triclinic material behavior}\label{sec:invs_tri}

A tuple of structural tensors for the triclinic group $\cC_1$ can be constructed through six second-order structural tensors given by $\bM^{\text{tri}}_{(i)}=\bu^{(i)}\otimes\bu^{(i)}$ (cf.~\cref{eq:C_comp}). The six $\bu^{(i)}$ are chosen by replacing $\be^{(i)}$ with $\bn_0^{(i)}$ in~\cref{eq:cu}. In this case, the choice of $\bn_0^{(i)}$ is generic given that $\bn_{0}^{(i)}\cdot\bn_{0}^{(j)}=\delta_{ij}$. Then, $\bn_0^{(1-3)}$ span the $\bbR^3$, and the tuple of structural tensors characterizes the triclinic group. The polyconvex integrity basis and its extension for triclinic anisotropy can then be obtained as discussed in~\cref{app:invs_tri_coord}. A polyconvex integrity basis for the triclinic symmetry group $\cC_1$ was first proposed by \textcite[Sec.~3.2]{Schroeder2008}.


\subsubsection{Monoclinic material behavior}\label{sec:invs_MON}

\paragraph{(A) Minimal integrity basis} The monoclinic symmetry group $\cC_{2}$ can be characterized with the two second-order structural tensors
\begin{equation}\label{eq:struct_mono_standard}
      \cS^{\text{mon}}=\big(\bMmonone,\,\bMmontwo\big)\qquad\text{with}\qquad  
      \begin{aligned}
          \bMmonone&=\bn_{0}^{(1)}\otimes\bn_{0}^{(1)}-\bn_{0}^{(2)}\otimes\bn_{0}^{(2)}\,,
        \\
        \bMmontwo&=\bn_{0}^{(1)}\otimes\bn_{0}^{(2)}-\bn_{0}^{(2)}\otimes\bn_{0}^{(1)}\,.
      \end{aligned}
\end{equation}
Together with the isotropic invariants $J_{1,2}^{\text{iso}}$ introduced in~\cref{sec:invs_iso}, the invariants
\begin{equation}\label{eq:invs_mono_unr}
\begin{aligned}
&{\mathrlap{\JoneMON=\tr\bC\cdot\big(\bMmontwo\big)^2\,,}\phantom{\JoneMON=\tr\bC\cdot\big(\bMmontwo\big)^2\,,}}\qquad 
    {\mathrlap{\JtwoMON=\tr\bC\cdot\bMmonone}{\phantom{ \JfourMON=\tr\bC^2\cdot\bMmonone}}}\,,
    \qquad
    \JthreeMON=\tr\bC\cdot\bMmonone\cdot\bMmontwo\,,
    \\ 
    &\phantom{\JoneMON=\tr\bC\cdot\big(\bMmontwo\big)^2\,,}\qquad   \JfourMON=\tr\bC^2\cdot\bMmonone\,,\qquad \JfiveMON=\tr\bC^2\cdot\bMmonone\cdot\bMmontwo\,, 
\end{aligned}
\end{equation}
form a minimal integrity basis for the monoclinic $\cC_2$ group.

\paragraph{(B) Polyconvex integrity basis} 

The monoclinic structural tensors in~\cref{eq:struct_mono_standard} are not p.s.d.\ and thus not suitable for the construction of polyconvex invariants. Instead, in the following, we consider 
\begin{equation}
      \widetilde{\cS}^{\text{mon}}=\big(\bMmononePC,\,\bMmontwoPC\big)\qquad\text{with}\qquad  
      \begin{aligned}
          \bMmononePC=&\bn_{0}^{(1)}\otimes\bn_{0}^{(1)}\,,
        \\
        \bMmontwoPC=&\big(\bn_{0}^{(1)}+\bn_{0}^{(2)}\big)\otimes\big(\bn_{0}^{(1)}+\bn_{0}^{(2)}\big)
        \\
        &+2\big(\bn_{0}^{(1)}\otimes\bn_{0}^{(1)}+\bn_{0}^{(2)}\otimes\bn_{0}^{(2)}\big)\,,
      \end{aligned}
\end{equation}
which is an alternative tuple of p.s.d.\ structural tensors for monoclinic anisotropy \parencite{olive2022}, see also \parencite[Sec.~A.3]{kalina2025}. This shows that, for some symmetry groups, identifying a suitable tuple of p.s.d.\ structural tensors is a first step towards the construction of a polyconvex integrity basis. Together with the isotropic invariants $I_{1-3}^{\text{iso}}$, the invariants
\begin{equation}\label{eq:invs_mon_pc}
\begin{aligned}
&{\mathrlap{\IoneMON=\tr\bC\cdot\bMmononePC\,,}\phantom{\IoneMON=\tr\bC\cdot\bMmononePC\,,}}\qquad 
    {\mathrlap{\ItwoMON=\tr\bC\cdot\bMmontwoPC}{\phantom{\IfourMON=\tr\bG\cdot\bMmononePC}}}\,,
    \qquad
   \IthreeMON=\tr\bC\cdot\big(\bMmontwoPC\big)^2\,,
    \\ 
    &\phantom{\IoneMON=\tr\bC\cdot\bMmononePC\,,}\qquad \IfourMON=\tr\bG\cdot\bMmononePC\,,\qquad \IfiveMON=\tr\bG\cdot\bMmontwoPC\,, 
\end{aligned}
\end{equation}
form a polyconvex integrity basis for the monoclinic $\cC_2$ group. This is because $I^{\text{mon}}_{1-5}$ and $I^{\text{iso}}_{1,2}$ can represent $J^{\text{iso}}_{1,2}$ and $J^{\text{mon}}_{1-5}$ as polynomials, cf.~App.~\ref{app:inv_relations_iso} and~App.~\ref{app:inv_relations_ti}. A polyconvex integrity basis for the monoclinic $\cC_2$ group was first proposed by \textcite[Sec.~3.2]{Schroeder2008}.

\paragraph{(C) Extension of the polyconvex integrity basis} Similar to~\cref{eq:invs_ti_2}, additional polyconvex invariants could be constructed based on the structural tensor $\bMmononePC$, cf.~Remark~\ref{rem:add_invs}.


\subsubsection{Rhombic material behavior}\label{sec:invs_RHO}

\paragraph{(A) Minimal integrity basis} 

For the rhombic symmetry group $\cD_{2}$, a single second-order structural tensor is sufficient to characterize the materials anisotropy, e.g.,
\begin{equation}\label{eq:struct_rhom_standard}
  \cS^{\text{rho}}=\big(\bMrho\big)\qquad\text{with}\qquad      \bMrho=\bn_{0}^{(1)}\otimes\bn_{0}^{(1)}-\bn_{0}^{(2)}\otimes\bn_{0}^{(2)}\,.
\end{equation}
Together with the isotropic invariants $J_{1-3}^{\text{iso}}$ introduced in~\cref{sec:invs_iso}, the invariants
\begin{equation}\label{eq:invs_rho_unr}
\begin{aligned}
&\JoneRHO=\tr\bC\cdot\bMrho\,,\qquad&&\JtwoRHO=\tr\bC\cdot\big(\bMrho\big)^2\,,
\\
&\JthreeRHO=\tr\bC^2\cdot\bMrho\,,&&\JfourRHO=\tr\bC^2\cdot\big(\bMrho\big)^2\,,
\end{aligned}    
\end{equation}
form a minimal integrity basis for the rhombic $\cD_2$ group.

\paragraph{(B) Polyconvex integrity basis}

The rhombic structural tensor in~\cref{eq:struct_rhom_standard} is not p.s.d.\ and thus not suitable for the construction of polyconvex invariants. Instead, in the following, we consider 
\begin{equation}
      \widetilde{\cS}^{\text{rho}}=\big(\bMrhoONE,\,\bMrhoTWO\big)\qquad\text{with}\qquad  
      \begin{aligned}
          \bMrhoONE=&\bn_{0}^{(1)}\otimes\bn_{0}^{(1)}\,,
        \\
        \bMrhoTWO=&\bn_{0}^{(2)}\otimes\bn_{0}^{(2)}\,,
      \end{aligned}
\end{equation}
which is an alternative tuple of p.s.d.\ structural tensors for rhombic anisotropy \parencite[Sec.~A.4]{kalina2025}. Together with the isotropic invariants $I_{1-3}^{\text{iso}}$, the invariants
\begin{equation}\label{eq:invs_rho_pc}
\begin{aligned}
\IoneRHO=\tr\bC\cdot\bMrhoONE\,,\qquad\ItwoRHO=\tr\bC\cdot\bMrhoTWO\,,
\\
\IthreeRHO=\tr\bG\cdot\bMrhoONE\,,\qquad\IfourRHO=\tr\bG\cdot\bMrhoTWO\,,
\end{aligned}
\end{equation}
form a polyconvex integrity basis for the rhombic $\cD_2$ group. This is because $I^{\text{rho}}_{1-4}$ and $I^{\text{iso}}_{1-3}$ can represent $J^{\text{rho}}_{1-4}$ and $J^{\text{iso}}_{1-3}$ as polynomials, cf.~\textcite[Sec.~5]{Schroeder2003}, App.~\ref{app:inv_relations_iso} and~App.~\ref{app:inv_relations_rho}. A polyconvex integrity basis for the rhombic $\cD_{2}$ group was first proposed by \textcite[Sec.~5]{Schroeder2003}.
 
\paragraph{(C) Extension of the polyconvex integrity basis} Similar to~\cref{eq:invs_ti_2}, additional polyconvex invariants could be constructed based on the structural tensor $\bMrhoONE$ and $\bMrhoTWO$~\parencite[Sec.~5]{Schroeder2003}, cf.~Remark~\ref{rem:add_invs}.

\subsection{On (in-)complete polyconvex integrity or functional bases for anisotropy}\label{app:complete_pc_bases}

{In this manuscript}, we have presented polyconvex integrity and functional bases for several material symmetry groups, including five crystal groups, one transversely isotropic group, and the isotropic group. However, six more crystal groups and one more transversely isotropic group are also relevant in hyperelastic constitutive modeling. To the best of our knowledge, no polyconvex integrity or functional bases have been proposed for these remaining symmetry groups. The formulation of polyconvex bases for \emph{all} symmetry groups entails challenges such as invariants that include higher-order powers of $\bC$, structural tensors that are not p.s.d., and complex fourth- and sixth-order structural tensors. As we demonstrated for tetragonal and cubic anisotropy {in~\cref{sec:inv_bases_group_symm}}, group symmetrization provides an alternative approach to the formulation of anisotropic invariants and may present an opportunity for the construction of further polyconvex integrity or functional bases. To conclude, for a large share of symmetry groups, no polyconvex integrity or functional basis is available. As a consequence, polyconvex constitutive models for such symmetry groups are based on incomplete sets of invariants. This leads to a loss of information, thus limiting the flexibility of the constitutive models. We provide a comprehensive discussion on such limitations of polyconvex constitutive modeling in~\textcite{klein2026b}.



\subsection{Relations between different invariant bases}\label{app:inv_relations}

We now provide functional relations between different invariant bases introduced {throughout this manuscript}. For each respective symmetry group, we provide 
\begin{enumerate}[label=(\roman*)]
\setlength\itemsep{0.0em}
\item a representation of the polyconvex invariants $I_i^{\square}$ in terms of the general invariants $J_i^{\square}$, to verify that the invariants $I_i^{\square}$ belong to the considered symmetry group, and
\item a representation of the general invariants $J_i^{\square}$ in terms of the polyconvex invariants $I_i^{\square}$, to investigate whether the latter form an integrity or functional basis for the considered symmetry group. 
\end{enumerate}
For (ii), we make use of the fact that we consider sets of general invariants $\bcI^{\square}:=(J_1^{\square},\,\dotsc,\,J_m^{\square})\in\bbR^m$ that form minimal integrity bases for each respective symmetry group. 
For some of the symmetry groups, these results are well-established, but we still provide them for the convenience of the reader. For triclinic anisotropy, both (i) and (ii) are given by linear relationships. Thus, we do not consider triclinic anisotropy in the following, and leave the calculations to the interested reader. We provide symbolic MATLAB code to verify the following calculations on the public GitHub repository \url{https://github.com/CPShub/sim-data}.


\subsubsection{Isotropic material behavior}\label{app:inv_relations_iso}

We consider the isotropic $\cK$ group, with the general invariants $J_{1-3}^{\text{iso}}$ as defined in~\cref{eq:invs_iso_standard} and the polyconvex invariants $I_{1-3}^{\text{iso}}$ as defined in~\cref{eq:invs_iso}. The following results can also be found in, e.g., \textcite[Eq.~(4.29)]{itskov2015} and \textcite[Eq.~(3.18)]{Schroeder2003}.

\paragraph{(i) Representation of polyconvex invariants in terms of general invariants}
\begin{itemize}
\setlength\itemsep{0.0em}
\item We have the trivial relation $\IoneISO=\JoneISO$.
\item The invariant $\ItwoISO$ can be expressed as
\begin{equation}
\begin{aligned}
\ItwoISO =\frac{1}{2}\Big(\big(\JoneISO\big)^2-\JtwoISO\Big)\,.
\end{aligned}
\end{equation}
%
\item The invariant $\IthreeISO$ can be expressed as
\begin{equation}
\begin{aligned}
  \IthreeISO=\frac{1}{6}\big(\JoneISO\big)^3-\frac{1}{2} \JoneISO\JtwoISO+\frac{1}{3}\JthreeISO \,.  
\end{aligned}
\end{equation}
\end{itemize}
This verifies that $I_{1-3}^{\text{iso}}$ are invariants of the isotropic $\cK$ group.

\paragraph{(ii) Representation of general invariants in terms of polyconvex invariants}
\begin{itemize}
\setlength\itemsep{0.0em}
%
\item We have the trivial relation $\JoneISO=\IoneISO$.
%
\item The invariant $\JtwoISO$ can be expressed as
\begin{equation}
\begin{aligned}
\JtwoISO =\big(\IoneISO\big)^2-2\ItwoISO\,.
\end{aligned}
\end{equation}
%
\item The invariant $\JthreeISO$ can be expressed as
\begin{equation}
\begin{aligned}
  \JthreeISO=\big({\IoneISO}\big)^3-3\IoneISO\ItwoISO+3\IthreeISO\,.
  \end{aligned}
\end{equation}
\end{itemize}
The invariants $J_{1-3}^{\text{iso}}$ can be expressed as polynomials of the invariants $I_{1-3}^{\text{iso}}$. Thus, the invariants $I_{1-3}^{\text{iso}}$ form an integrity basis for the isotropic $\cK$ group.


\subsubsection{Transversely isotropic material behavior}\label{app:inv_relations_ti}

We consider the transversely isotropic $\cD_{\infty}$ group, with the general invariants $J_{1,2}^{\text{ti}}$ as defined in~\cref{eq:invs_ti_unr} and the polyconvex invariants $I_{1,2}^{\text{ti}}$ as defined in~\cref{eq:invs_ti_pc}. The following results can also be found in, e.g., \textcite[Eq.~(3.46)]{Schroeder2003}.

\paragraph{(i) Representation of polyconvex invariants in terms of general invariants}
\begin{itemize}
\setlength\itemsep{0.0em}
\item We have the trivial relation $\IoneTI=\JoneTI$.
%
\item The invariant $\ItwoTI$ can be expressed as
\begin{equation}
\begin{aligned}
\ItwoTI =\frac{1}{2}\Big(\big(\JoneISO\big)^2-\JtwoISO\Big)-\JoneISO\JoneTI+\JtwoTI\,.
\end{aligned}
\end{equation}
\end{itemize}
This verifies that $I_{1,2}^{\text{ti}}$ are invariants of the transversely isotropic $\cD_{\infty}$ group.

\paragraph{(ii) Representation of general invariants in terms of polyconvex invariants}
\begin{itemize}
\setlength\itemsep{0.0em}
%
\item We have the trivial relation $\JoneTI=\IoneTI$.
%
\item The invariant $\JtwoTI$ can be expressed as
\begin{equation}
\begin{aligned}
\JtwoTI =\IoneISO\IoneTI-\ItwoISO+\ItwoTI\,.
\end{aligned}
\end{equation}
\end{itemize}
The invariants $I^{\text{ti}}_{1,2}$ and $I^{\text{iso}}_{1-3}$ can represent $J^{\text{iso}}_{1-3}$ and $J^{\text{ti}}_{1,2}$ as polynomials, see also~App.~\ref{app:inv_relations_iso}. Thus, $I^{\text{ti}}_{1,2}$ and $I^{\text{iso}}_{1-3}$ form an integrity basis for the transversely isotropic $\cD_{\infty}$ group.


\subsubsection{Monoclinic material behavior}\label{app:inv_relations_mon}

We consider the monoclinic $\cC_2$ group, with the general invariants $J_{1-5}^{\text{mon}}$ as defined in~\cref{eq:invs_mono_unr} and the polyconvex invariants $I_{1-5}^{\text{mon}}$ as defined in~\cref{eq:invs_mon_pc}.

\paragraph{(i) Representation of polyconvex invariants in terms of general invariants}
\begin{itemize}
\setlength\itemsep{0.0em}
\item The invariant $\IoneMON$ can be expressed as
\begin{equation}
\begin{aligned}
\IoneMON =\frac{1}{2}\Big(\JtwoMON-\JoneMON\Big)\,.
\end{aligned}
\end{equation}
\item The invariant $\ItwoMON$ can be expressed as
\begin{equation}
\begin{aligned}
\ItwoMON = \JthreeMON-3\JoneMON\,.
\end{aligned}
\end{equation}
%
\item The invariant $\IthreeMON$ can be expressed as
\begin{equation}
\begin{aligned}
\IthreeMON =6\JthreeMON-10\JoneMON\,.
\end{aligned}
\end{equation}
%
\item The invariant $\IfourMON$ can be expressed as
\begin{equation}
\begin{aligned}
\IfourMON =\frac{1}{4}\Big(\big(\JoneISO\big)^2-\JtwoISO\Big)+\frac{1}{2}\Big(\JfourMON-\JoneISO\JtwoMON\Big)+\frac{1}{8}\Big(\big(\JtwoMON\big)^2+\big(\JthreeMON\big)^2-\big(\JoneMON\big)^2\Big)\,.
\end{aligned}
\end{equation}
%
\item The invariant $\IfiveMON$ can be expressed as
\begin{equation}
\begin{aligned}
\IfiveMON =\JfiveMON-\JoneISO\JthreeMON+\frac{3}{2}\Big(\big(\JoneISO\big)^2-\JtwoISO\Big)
+\frac{3}{4}\Big(\big(\JthreeMON\big)^2-\big(\JoneMON\big)^2+\big(\JtwoMON\big)^2\Big)\,.
\end{aligned}
\end{equation}
\end{itemize}
This verifies that $I_{1-5}^{\text{mon}}$ are invariants of the monoclinic $\cC_2$ group.

\paragraph{(ii) Representation of general invariants in terms of polyconvex invariants}
\begin{itemize}
\setlength\itemsep{0.0em}
%
\item The invariant $\JoneMON$ can be expressed as
\begin{equation}
\begin{aligned}
\JoneMON =\frac{1}{8}\IthreeMON-\frac{3}{4}\ItwoMON\,.
\end{aligned}
\end{equation}
%
\item The invariant $\JtwoMON$ can be expressed as
\begin{equation}
\begin{aligned}
\JtwoMON =2\IoneMON-\frac{3}{4}\ItwoMON+\frac{1}{8}\IthreeMON\,.
\end{aligned}
\end{equation}
%
\item The invariant $\JthreeMON$ can be expressed as
\begin{equation}
\begin{aligned}
\JthreeMON =\frac{3}{8}\IthreeMON-\frac{5}{4}\ItwoMON\,.
\end{aligned}
\end{equation}
%
\item The invariant $\JfourMON$ can be expressed as
\begin{equation}
\begin{aligned}
\JfourMON =&-\ItwoISO-\big(\IoneMON\big)^2+2\Big(\IoneISO\IoneMON+\IfourMON\Big)+\frac{1}{8}\Big(\IoneISO\IthreeMON-\IoneMON\IthreeMON\Big)
\\
&+\frac{3}{4}\Big(\IoneMON\ItwoMON-\IoneISO\ItwoMON\Big)-\frac{9}{256}\big(\IthreeMON\big)^2+\frac{15}{64}\ItwoMON\IthreeMON-\frac{25}{64}\big(\ItwoMON\big)^2\,.
\end{aligned}
\end{equation}
%
\item The invariant $\JfiveMON$ can be expressed as
\begin{equation}
\begin{aligned}
\JfiveMON =&\IfiveMON-3\Big(\big(\IoneMON\big)^2+\ItwoISO\Big)+\frac{9}{4}\IoneMON\ItwoMON
-\frac{5}{4}\IoneISO\ItwoMON+\frac{3}{8}\Big(\IoneISO\IthreeMON-\IoneMON\IthreeMON\Big)
\\
&+\frac{45}{64}\ItwoMON\IthreeMON-\frac{75}{64}\big(\ItwoMON\big)^2
-\frac{27}{256}\big(\IthreeMON\big)^2
\end{aligned}
\end{equation}
\end{itemize}
The invariants $I^{\text{mon}}_{1-5}$ and $I^{\text{iso}}_{1,2}$ can represent $J^{\text{iso}}_{1,2}$ and $J^{\text{mon}}_{1-5}$ as polynomials, see also~App.~\ref{app:inv_relations_iso}. Thus, $I^{\text{mon}}_{1-5}$ and $I^{\text{iso}}_{1,2}$ form an integrity basis for the monoclinic $\cC_2$ group.


\subsubsection{Rhombic material behavior}\label{app:inv_relations_rho}

We consider the rhombic $\cD_2$ group, with the general invariants $J_{1-4}^{\text{rho}}$ as defined in~\cref{eq:invs_rho_unr} and the polyconvex invariants $I_{1-4}^{\text{rho}}$ as defined in~\cref{eq:invs_rho_pc}. Similar results can be found in \textcite[Eq.~(3.46)]{Schroeder2003}, where a different structural tensor for the general invariants is employed.

\paragraph{(i) Representation of polyconvex invariants in terms of general invariants}
\begin{itemize}
\setlength\itemsep{0.0em}
\item The invariant $\IoneRHO$ can be expressed as
\begin{equation}
\begin{aligned}
\IoneRHO =\frac{1}{2}\Big(\JoneRHO+\JtwoRHO\Big)\,.
\end{aligned}
\end{equation}
\item The invariant $\ItwoRHO$ can be expressed as
\begin{equation}
\begin{aligned}
\ItwoRHO =\frac{1}{2}\Big(\JtwoRHO-\JoneRHO\Big)\,.
\end{aligned}
\end{equation}
%
\item The invariant $\IthreeRHO$ can be expressed as
\begin{equation}
\begin{aligned}
\IthreeRHO =\frac{1}{2}\Big(\big(\JoneISO\big)^2-\JoneISO\JoneRHO-\JoneISO\JtwoRHO-\JtwoISO+\JthreeRHO+\JfourRHO\Big)\,.
\end{aligned}
\end{equation}
%
\item The invariant $\IfourRHO$ can be expressed as
\begin{equation}
\begin{aligned}
\IfourRHO =\frac{1}{2}\Big(\big(\JoneISO\big)^2+\JoneISO\JoneRHO-\JoneISO\JtwoRHO-\JtwoISO-\JthreeRHO+\JfourRHO\Big)\,.
\end{aligned}
\end{equation}
\end{itemize}
This verifies that $I_{1-4}^{\text{rho}}$ are invariants of the rhombic $\cD_2$ group.

\paragraph{(ii) Representation of general invariants in terms of polyconvex invariants}
\begin{itemize}
\setlength\itemsep{0.0em}
%
\item The invariant $\JoneRHO$ can be expressed as
\begin{equation}
\begin{aligned}
\JoneRHO =\IoneRHO-\ItwoRHO\,.
\end{aligned}
\end{equation}
%
\item The invariant $\JtwoRHO$ can be expressed as
\begin{equation}
\begin{aligned}
\JtwoRHO =\IoneRHO+\ItwoRHO\,.
\end{aligned}
\end{equation}
%
\item The invariant $\JthreeRHO$ can be expressed as
\begin{equation}
\begin{aligned}
\JthreeRHO =\IoneISO\IoneRHO-\IoneISO\ItwoRHO+\IthreeRHO-\IfourRHO\,.
\end{aligned}
\end{equation}
%
\item The invariant $\JfourRHO$ can be expressed as
\begin{equation}
\begin{aligned}
\JfourRHO =\IoneISO\IoneRHO+\IoneISO\ItwoRHO-2\ItwoISO+\IthreeRHO+\IfourRHO\,.
\end{aligned}
\end{equation}
\end{itemize}
The invariants $I^{\text{rho}}_{1-4}$ and $I^{\text{iso}}_{1-3}$ can represent $J^{\text{rho}}_{1-4}$ and $J^{\text{iso}}_{1-3}$ as polynomials, see also~App.~\ref{app:inv_relations_iso}. Thus, {$I^{\text{rho}}_{1-4}$} and $I^{\text{iso}}_{1-3}$ form an integrity basis for the rhombic $\cD_2$ group.


\subsubsection{Tetragonal material behavior}\label{app:inv_relations_tet}

We consider the tetragonal $\cD_4$ group, with the general invariants $J_{1-5}^{\text{tet}}$ as defined in~\cref{eq:invs_tet_brain} and the polyconvex invariants $I_{1-5}^{\text{tet}}$ as defined in~\cref{eq:tet_invs_pc_basis}.

\paragraph{(i) Representation of polyconvex invariants in terms of general invariants}

\begin{itemize}
\setlength\itemsep{0.0em}
\item We have the trivial relations $\IoneTET=\JoneTET$ and $\ItwoTET=\JtwoTET$.
\item The invariant $\IthreeTET$ can be expressed as
\begin{equation}
\begin{aligned}
\IthreeTET =\JthreeTET + \big(\JoneISO\big)^2-\JoneTET\JoneISO-\JtwoISO\,.
\end{aligned}
\end{equation}
%
\item The invariant $\IfourTET$ can be expressed as
\begin{equation}
\begin{aligned}
  \IfourTET=&\frac{1}{2}\Big(\big(\JoneISO\big)^4+\big(\JtwoISO\big)^2\Big)
  -\big(\JoneISO\big)^3\JoneTET-\big(\JoneISO\big)^2\JtwoISO+\big(\JoneISO\big)^2\JtwoTET
  \\&+\big(\JoneISO\big)^2\JthreeTET+\JoneISO\JtwoISO\JoneTET-2\JoneISO\JfourTET-\JtwoISO\JthreeTET+\JfiveTET\,.
  \end{aligned}
\end{equation}
%
\item The invariant $\IfiveTET$ can be expressed as
\begin{equation}
\begin{aligned}
\IfiveTET = \ItwoTET+\IfourTET+2\bC:\bMtet:\bG\,,
\end{aligned}
\end{equation}
with
\begin{equation}
\begin{aligned}
\bC:\bMtet:\bG=\frac{1}{2}\Big(\big(\JoneISO\big)^2\JoneTET-\JtwoISO\JoneTET\Big)
-\JoneISO\JtwoTET+\JfourTET=\frac{1}{2}\big(\IfiveTET-\ItwoTET-\IfourTET\big)\,.
\end{aligned}
\end{equation}
\end{itemize}
This verifies that $I_{1-5}^{\text{tet}}$ are invariants of the tetragonal $\cD_4$ group.

\paragraph{(ii) Representation of general invariants in terms of polyconvex invariants}
\begin{itemize}
\setlength\itemsep{0.0em}
\item We have the trivial relations $\JoneTET=\IoneTET$ and $\JtwoTET=\ItwoTET$.
\item The invariant $\JthreeTET$ can be expressed as
\begin{equation}
\begin{aligned}
\JthreeTET =\IoneISO\IoneTET-2\ItwoISO+\IthreeTET\,.
\end{aligned}
\end{equation}
%
\item The invariant $\JfourTET$ can be expressed as
\begin{equation}
\begin{aligned}
  \JfourTET=\IoneISO\ItwoTET-\ItwoISO\IoneTET+\bC:\bMtet:\bG\,.
\end{aligned}
\end{equation}
%
\item The invariant $\JfiveTET$ can be expressed as
\begin{equation}
\begin{aligned}
\JfiveTET = \big(\IoneISO\big)^2\ItwoTET+\IfourTET+2\Big(\big(\ItwoISO\big)^2-\IoneISO\ItwoISO\IoneTET+\IoneISO\bC:\bMtet:\bG-\ItwoISO\IthreeTET\Big)\,.
\end{aligned}
\end{equation}
\end{itemize}
The invariants $I^{\text{tet}}_{1-5}$ and $I^{\text{iso}}_{1-3}$ can represent $J^{\text{tet}}_{1-5}$ and $J^{\text{iso}}_{1-3}$ as polynomials, see also~App.~\ref{app:inv_relations_iso}. Thus, $I^{\text{tet}}_{1-5}$ and $I^{\text{iso}}_{1-3}$ form an integrity basis for the tetragonal $\cD_4$ group.


\subsubsection{Cubic material behavior}\label{app:inv_relations_cub}

We consider the cubic $\cO$ group, with the general invariants $J_{1-6}^{\text{cub}}$ as defined in~\cref{eq:invs_cub_karl} and the polyconvex invariants $I_{1-6}^{\text{cub}}$ as defined in~\cref{eq:cub_invs_pc_basis}.

\paragraph{(i) Representation of polyconvex invariants in terms of general invariants}

\begin{itemize}
\setlength\itemsep{0.0em}
\item We have the trivial relations $\IoneCUB=\JoneCUB$ and $\ItwoCUB=\JtwoCUB$.
%
\item The invariant $\IthreeCUB$ can be expressed as
\begin{equation}
\begin{aligned}
& \IthreeCUB =\IoneCUB+\IfourCUB+2\bC:\bMcub:\bG\,,
\end{aligned}
\end{equation}
with
\begin{equation}
\begin{aligned}
\bC:\bMcub:\bG=\JthreeCUB-\JoneISO\JoneCUB+\frac{1}{2}\Big(\big(\JoneISO\big)^3-\JoneISO\JtwoISO\Big)=\frac{1}{2}\Big(\IthreeCUB-\IoneCUB-\IfourCUB\Big)\,.
\end{aligned}
\end{equation}
%
\item The invariant $\IfourCUB$ can be expressed as
\begin{equation}
\begin{aligned}
\IfourCUB=\JfourCUB+\big(\JoneISO\big)^2\JoneCUB-2\JoneISO\JthreeCUB+\frac{1}{2}\big(\JoneISO\big)^2\JtwoISO-\frac{1}{4}\Big(\big(\JtwoISO\big)^2+\big(\JoneISO\big)^4\Big)\,.
\end{aligned}
\end{equation}
%
\item The invariant $\IfiveCUB$ can be expressed as
\begin{equation}\label{eq:cub_invs_pc_basis_5}
\begin{aligned}
\IfiveCUB=\tr\left(\bMcub:\bC\right)^4+\IfourCUB+2\sum_{i=1}^3\left[\bn_{0}^{(i)}\otimes\bn_{0}^{(i)}:\bC\right]^2\left[\bn_{0}^{(i)}\otimes\bn_{0}^{(i)}:\bG\right]\,,
\end{aligned}
\end{equation}
with
\begin{equation}\label{eq:cub_invs_pc_basis_6}
\begin{aligned}
\tr\left(\bMcub:\bC\right)^4&=\frac{1}{2}\big(\JoneCUB\big)^2-\JoneCUB\big(\JoneISO\big)^2+\frac{1}{6}\big(\JoneISO\big)^4+\frac{4}{3}\JtwoCUB\JoneISO
\\
&=\frac{1}{6}\big(\IoneISO\big)^4-\big(\IoneISO\big)^2\IoneCUB+\frac{4}{3}\IoneISO\ItwoCUB+\frac{1}{2}\big(\IoneCUB\big)^2\,,
\end{aligned}
\end{equation}
and
\begin{equation}
\begin{aligned}
\sum_{i=1}^3\left[\bn_{0}^{(i)}\otimes\bn_{0}^{(i)}:\bC\right]^2\left[\bn_{0}^{(i)}\otimes\bn_{0}^{(i)}:\bG\right]&=\JfiveCUB-\JtwoCUB\JoneISO+\frac{1}{2}\Big(\JoneCUB\big(\JoneISO\big)^2-\JoneCUB\JtwoISO\Big)
\\
&=\frac{1}{2}\Big(\IfiveCUB-\tr\left(\bMcub:\bC\right)^4-\IfourCUB\Big)
\end{aligned}
\end{equation}
%
\item The invariant $\IsixCUB$ can be expressed as
\begin{equation}\label{eq:cub_invs_pc_basis_8}
\begin{aligned}
\IsixCUB=&-\JsixCUB\JoneISO+\frac{3}{4}\big(\JoneISO\big)^4\JtwoISO-\frac{7}{24}\big(\JoneISO\big)^6+\frac{9}{8}\big(\JoneISO\big)^4\JoneCUB-\frac{2}{3}\big(\JoneISO\big)^3\JtwoCUB-2\big(\JoneISO\big)^3\JthreeCUB
\\
&-\frac{5}{8}\big(\JoneISO\big)^2\big(\JtwoISO\big)^2-\frac{3}{4}\big(\JoneISO\big)^2\JtwoISO\JoneCUB+\big(\JoneISO\big)^2\JfourCUB+\frac{3}{2}\big(\JoneISO\big)^2\JfiveCUB+\JoneISO\JtwoISO\JthreeCUB
\\
&+\frac{1}{4}\big(\JtwoISO\big)^3-\frac{1}{8}\big(\JtwoISO\big)^2\JoneCUB-\frac{3}{4}\JtwoISO\JfourCUB+\frac{1}{2}\JtwoISO\JfiveCUB-\frac{1}{12}\big(\JthreeISO\big)^2
\\
&-\frac{1}{3}\JthreeISO\JtwoCUB+\frac{1}{2}\JthreeISO\JthreeCUB-\frac{1}{4}\JoneCUB\JfourCUB+\frac{1}{4}\big(\JthreeCUB\big)^2\,.
\end{aligned}
\end{equation}
\end{itemize}
This verifies that $I_{1-6}^{\text{cub}}$ are invariants of the cubic $\cO$ group.

\paragraph{(ii) Representation of general invariants in terms of polyconvex invariants}

\begin{itemize}
\setlength\itemsep{0.0em}
\item We have the trivial relations $\JoneCUB=\IoneCUB$ and $\JtwoCUB=\ItwoCUB$.
%
\item The invariant $\JthreeCUB$ can be expressed as
\begin{equation}
\begin{aligned}
& \JthreeCUB =\IoneISO\IoneCUB-\IoneISO\ItwoISO+\bC:\bMcub:\bG\,.
\end{aligned}
\end{equation}
\item The invariant $\JfourCUB$ can be expressed as
\begin{equation}
\begin{aligned}
\JfourCUB=\Big(\IoneISO\Big)^2\IoneCUB-2\Big(\IoneISO\Big)^2\ItwoISO+2\IoneISO\bC:\bMcub:\bG+\Big(\ItwoISO\Big)^2+\IfourCUB\,.
\end{aligned}
\end{equation}
%
\item The invariant $\JfiveCUB$ can be expressed as
\begin{equation}
\begin{aligned}
\JfiveCUB=\IoneISO\ItwoCUB-\ItwoISO\IoneCUB+\sum_{i=1}^3\left[\bn_{0}^{(i)}\otimes\bn_{0}^{(i)}:\bC\right]^2\left[\bn_{0}^{(i)}\otimes\bn_{0}^{(i)}:\bG\right]\,.
\end{aligned}
\end{equation}
%
%
\item The invariant $\JsixCUB$ can be expressed as
\begin{equation}\label{eq:Jsix_representation}
\begin{aligned}
\JsixCUB&=\big(\IoneISO\big)^2\ItwoCUB-\ItwoISO\bC:\bMcub:\bG+2\Big(\IoneISO\sum_{i=1}^3\left[\bn_{0}^{(i)}\otimes\bn_{0}^{(i)}:\bC\right]^2\left[\bn_{0}^{(i)}\otimes\bn_{0}^{(i)}:\bG\right]
\\
&-\IoneISO\ItwoISO\IoneCUB\Big)+\frac{1}{2}\Big(3\IthreeISO\IoneCUB-\big(\IoneISO\big)^2\IthreeISO\Big)+\frac{1}{4}\Big(3\IoneISO\big(\ItwoISO\big)^2+\IoneISO\IfourCUB\Big)
\\
&+\frac{1}{\IoneISO}\Bigg[-\IthreeISO\ItwoCUB-\IsixCUB-\ItwoISO\sum_{i=1}^3\left[\bn_{0}^{(i)}\otimes\bn_{0}^{(i)}:\bC\right]^2\left[\bn_{0}^{(i)}\otimes\bn_{0}^{(i)}:\bG\right]
\\
&+\frac{1}{2}\Big(3\ItwoISO\IfourCUB-\Big(\ItwoISO\Big)^3+3\IthreeISO\bC:\bMcub:\bG\Big)+\frac{1}{4}\Big(\big(\bC:\bMcub:\bG\big)^2
\\
&+\big(\ItwoISO\big)^2\IoneCUB-3\big(\IthreeISO\big)^2-\IoneCUB\IfourCUB\Big)\Bigg]
\end{aligned}
\end{equation}
\end{itemize}
The invariants $I^{\text{cub}}_{1-6}$ and $I^{\text{iso}}_{1-3}$ can represent $J^{\text{cub}}_{1-6}$ and $J^{\text{iso}}_{1-3}$, see also~App.~\ref{app:inv_relations_iso}. Thereby, $I^{\text{cub}}_{1-5}$ can be formulated as polynomials of $I^{\text{cub}}_{1-5}$ and $I^{\text{iso}}_{1-3}$. However, the representation of $\JsixCUB$ in terms of $I^{\text{cub}}_{1-6}$ and $I^{\text{iso}}_{1-3}$ is not a polynomial, as it includes a division by $\IoneISO$, cf.~\cref{eq:Jsix_representation}. Thus, $I^{\text{cub}}_{1-6}$ and $I^{\text{iso}}_{1-3}$ form only a functional basis but not an integrity basis for the cubic $\cO$ group. Note that $\IoneISO>0$, which means that \cref{eq:Jsix_representation} is well-defined.

\renewcommand*{\bibfont}{\footnotesize}
\printbibliography
\end{document}